\documentclass[11pt]{article}

\usepackage{appendix}
\usepackage{abstract}
\usepackage{float}
\usepackage{setspace} 
\usepackage{newpxtext} % Provides Palatino font shapes, including small caps italic

\usepackage{amsthm}
\usepackage{nicefrac}
\usepackage{comment}
\usepackage{mathtools}

\usepackage[square,numbers]{natbib}
\newcommand{\parencite}{\citep}

\usepackage{caption}
\usepackage[skip=2pt]{subcaption}
\usepackage{algpseudocode}
\usepackage{algorithm}
\usepackage{placeins}

\usepackage[sc]{mathpazo}

\usepackage{enumitem}
\setlist[itemize]{noitemsep, topsep=3pt}

\newenvironment{small_align}
{\par\vspace{-14pt}\footnotesize\csname align\endcsname}
 {\csname endalign\endcsname\par\vspace{-5pt}\ignorespacesafterend}
 
\newenvironment{small_align*}
{\par\vspace{-14pt}\footnotesize\expandafter\csname align*\endcsname}
 {\expandafter\csname endalign*\endcsname\par\vspace{-5pt}\ignorespacesafterend}

\newenvironment{small_equation}
{\par\vspace{-5pt}\footnotesize\equation}
{\endequation\par\vspace{-1pt}\ignorespacesafterend}

\newenvironment{small_equation*}
{\par\vspace{-5pt}\footnotesize\csname equation*\endcsname}
{\csname endequation*\endcsname\par\vspace{-1pt}\ignorespacesafterend}

\usepackage{thmtools}

\declaretheoremstyle[
  spaceabove=3pt, % Space before the proof
  spacebelow=3pt  % Space after the proof
]{tightproof}

\declaretheoremstyle[
    headfont=\bfseries,      % Bold "Theorem 1"
    notefont=\bfseries,      % Bold the optional name in brackets
    notebraces={(}{)},
    bodyfont=\itshape,       % Italic body text
    spaceabove=5pt,          % Tighter space above
    spacebelow=5pt,          % Tighter space below
    postheadspace=1em        % Space after the "Theorem 1." label
]{tightstyle}

\declaretheorem[name=Theorem, numberwithin=section]{theorem}
\declaretheorem[style=tightstyle, name=Lemma, numberwithin=section]{lemma}
\declaretheorem[style=tightstyle, name=Remark, numberwithin=section]{remark}

\declaretheorem[style=tightstyle, name=Proposition, numberwithin=section]{proposition}

\BeforeBeginEnvironment{proof}{\smallskip}

\AfterEndEnvironment{proof}{\smallskip}

\newcommand\restr[2]{{% we make the whole thing an ordinary symbol
  \left.\kern-\nulldelimiterspace % automatically resize the bar with \right
  #1 % the function
  \littletaller % pretend it's a little taller at normal size
  \right|_{#2} % this is the delimiter
  }}
\newcommand{\conv}{\operatorname{conv}} % 
\newcommand{\epi}{\operatorname{epi}}
\newcommand{\littletaller}{\mathchoice{\vphantom{\big|}}{}{}{}}

\newcommand{\tmp}[1]{#1}

\usepackage{authblk}
\title{On the Delay-Constrained Maximum Concurrent Flow Problem}

\author[1]{Walid Ben-Ameur \texttt{(walid.benameur@telecom-sudparis.eu)}}
\author[2,3]{Guillaume Beraud-Sudreau  (corresponding author \texttt{guillaume.beraud.sudreau@huawei.com})}
\author[3]{Herv\'e Kerivin \texttt{(herve.kerivin@uca.fr)}}
\author[2]{Sebastien Martin \texttt{(sebastien.martin@huawei.com)}}
\affil[1]{SAMOVAR, Télécom SudParis, Institut Polytechnique de Paris,{\scriptsize 19 place Marguerite Perey
91120 Palaiseau, France}}
\affil[2]{Huawei Technologies Ltd., Paris Research Center, {\scriptsize  18 Quai du Point du Jour, 92100 Boulogne-Billancourt, France}}
\affil[3]{LIMOS, Université Clermont Auvergne, CNRS, {\scriptsize  1 Rue de la Chebarde, 63178 Aubière, France}}

\date{March 2025}

\begin{document}
%\Large

\maketitle

%\begin{keywords}
%~OR in telecommunications, Nonlinear programming, Global Optimization
%\end{keywords}

\bigskip

\begin{abstract}
~Real-time services, such as VoIP and large-scale neural network training, require strict transmission delay guarantees. While routing under hop constraints is tractable, real-world delays increase sharply with equipment load, typically modeled using the M/M/1 queuing function where delay is inversely proportional to available bandwidth. We investigate the resulting Delay-Constrained Maximum Concurrent Flow (DCMCF) problem, which seeks to maximize the minimum throughput across all commodities. The problem's complexity stems from the conditional and non-linear nature of the delay constraints, which are active only along the specific paths used by the flow. We prove that DCMCF is strongly NP-hard, even for single-source/single-destination instances. To address the inherent non-convexity of the problem, we introduce a new convex relaxation expressed through second-order cone constraints, obtained from the convex envelope of a function representing the conditional delay associated with a single arc of a given path. 
The relaxation is shown to outperform existing formulations based on disjunctive programming. 
Leveraging this result, we develop a polynomial-time approximation algorithm with a provable performance guarantee and present numerical experiments demonstrating the effectiveness of the proposed approach.
\end{abstract}

\section{Introduction}

Over the past two decades, the expansion of cloud computing has intensified the demand for predictable end-to-end latency in telecommunications. This requirement is now critical across the spectrum, from consumer-facing applications like gaming and live streaming to large-scale enterprise operations, such as the distributed training of large language models across multiple data centers \parencite{dong2020tina}.

This requirement can be approached as a traffic-engineering problem, which is formulated as a multi-commodity flow problem, where the telecommunication network is represented by a directed graph $D = (V,A)$ and end-to-end communications are represented as commodities \parencite{wang2008}. 
The terms {vertex} and {node} are used interchangeably to denote an element of $V$, whereas {arc} and {link} refer to an element of $A$.
 An arc $a$ directed from $u$ to $v$ is denoted by $(u,v)$ and has a capacity $c_a>0$ which could also be denoted $c_{(u,v)}$. We consider a set $K$ of commodities with, for each commodity $k \in K$, a source $s^k$, a destination (or target) $t^k$ and a size (or volume) $b^k$. We also use $P^k$ to denote a given nonempty subset of paths from $s^k$ to $t^k$.  For $p \in P^k$, $x^k_p$  is the proportion of the commodity $k \in K$ carried through the path $p$.
A given value $d^k$   represents the maximum allowed delay for commodity $k$. 
The delay through an arc $a$, denoted by $d_a$, is given by the average delay function $d_a = \nicefrac{1}{(c_a - y_a)}$ following from the classical  M/M/1 queuing model \parencite{kleinrock1975theory}, where $y_a$ denotes the flow (or load) of $a$. 
Consequently, for any path $p \in P^k$ carrying positive flow (i.e., if $x^k_p >0$), the total delay through $p$ should be less than or equal to $d^k$.  
The Delay-Constrained Maximum {Concurrent} Flow problem (denoted {DCMCF}) can then be expressed as follows: 
\begin{small_align}
  \mbox{DCMCF:} \qquad & \max \gamma & \nonumber \\
    & {\sum_{p \in P^k} x^k_p  } = \gamma & \forall k \in K 
    \label{eq:dcmf:gamma}\\
    & y_a \geq \sum_{k \in K, p \in P^k \mid a \in p} x^k_p b^k  & \forall a \in A \label{eq:dcmf:xa}\\
    & \sum_{a \in p} \frac{1}{c_a - y_a} \leq d^k  \qquad \text{if } x^k_p > 0   & \forall k \in K, p \in P^k  \label{eq:dcmf:delay}\\
    & y_a \leq c_a                         & \forall a \in A \label{eq:dcmf:ca}\\
    & x^k_p \geq 0                      & \forall k \in K, p \in P^k. \label{eq:dcmf:pos}
\end{small_align}
Observe that $\gamma$ represents the throughput factor, or maximum proportion of each commodity that can be carried through the network. Constraints \eqref{eq:dcmf:gamma} establish the connection between $\gamma$ and the $x^k_p$ proportions ; the equality sign used in this formulation could be equivalently replaced by an inequality without changing the optimal solution of this problem. Constraints \eqref{eq:dcmf:xa} relate the flow on an arc $a\in A$ to the sum of flows on paths containing $a$.   
The delay constraints are stated  in \eqref{eq:dcmf:delay}. Observe that each Constraint \eqref{eq:dcmf:delay} is active if and only if the corresponding path is used (i.e., if $x^k_p > 0$).
Constraints \eqref{eq:dcmf:ca} introduce capacity limitations; note that, because of the delay constraint, the load of each arc is always strictly lower than its capacity.
Observe that the delay along a path $p \in P^k$ is at least $\sum_{a \in p} \frac{1}{c_a}$. We may therefore assume that $\sum_{a \in p} \frac{1}{c_a} < d^k$ for each path $p \in P^k$, since any path violating this condition can be discarded.\\
After multiplication by $x^k_p$, \eqref{eq:dcmf:delay} can be reformulated as:
\begin{small_equation}
    \sum_{a \in p} \frac{x^k_p}{c_a - y_a} \leq x^k_p d^k \qquad \forall k \in K, p \in P^k. \label{eq:dcmf:delay_time_xkp}
\end{small_equation}
It is worth highlighting that, while this reformulation of the delay constraint incorporates the path-usage condition into a differentiable analytical form, the resulting constraint is not convex.

\medskip

We can now introduce some relevant notation for the convexity analysis of functions. Given a bounded real-valued function $f$ defined on 
a compact nonempty set $\mathcal{X}$,  
the convex envelope of
 $f$ over  $\mathcal{X}$ is the function $\check{f}: \mathcal{X}\ni x
\mapsto \check f(x) = \inf \{u: (x,u) \in \conv(\epi (f)) \}$, where $\epi (f) =
\{(x,u): f(x) \le u\}$ and $\conv(.)$ of a set represents its convex hull. In other words, $\check{f}$ is the highest convex under-estimator of $f$ over $\mathcal{X}$. 

If this is not clear from the context, the subscript $\mathcal{X}$ can be added to emphasize that we are considering the convex envelope over $\mathcal{X}$, thus leading to the notation $\check f_{\mathcal{X}}$. 

One can also express $\check{f}(x)$  by considering a convex combination over   $\mathcal{X}$ representing $x$ and  minimizing the combination of $f$'s values: 
\begin{small_align}  \check f_{\mathcal{X}}(x) =  & \inf\limits_{\lambda_i} 
 \sum\limits_{v_i \in \mathcal{X}} \lambda_i 
 f(v_i)  \nonumber \\
  \sum_{v_i \in \mathcal{X}} \lambda_i v_i=& x,\quad \sum\limits_{v_i \in \mathcal{X}} \lambda_i = 1, \lambda_i
\ge 0.  \nonumber
\end{small_align}
%\end{comment}
\subsection{Related work}

\subsubsection*{Maximum concurrent flow}  
The maximum concurrent flow problem without delay constraints can be formulated as a linear program and, in principle, solved using standard optimization tools \parencite{ahuja1988network}. Nevertheless, even this version—without delay constraints—is known to give rise to  challenging computational problems (see, e.g., \parencite{B01,BR02}). This led many researchers to design  approximation algorithms for the problem \parencite{BR02,Fleischer00, GargK07,KAR08,LR99,MS86}. Furthermore, several decomposition methods are proposed in~\parencite{Bauguion} to improve computational efficiency. In the special case where all commodities share a common source, a strongly polynomial-time combinatorial algorithm is also presented in~\parencite{Bauguion}.\\
Regarding the theoretical significance of the maximum concurrent flow problem, it can be interpreted as the dual of the linear programming relaxation of the sparsest cut problem. A sparsest cut is defined as a cut that minimizes the ratio between the capacity of the cut and the total demand separated by it. Many approximation algorithms for the sparsest cut problem are based on solutions to the maximum concurrent flow problem (see, e.g., \parencite{S97}).

\subsubsection*{Flows with delay considerations}
When the delay on each link is constant and identical across all links, incorporating delay constraints is equivalent to imposing a bound on the number of links in the selected path. This  can be handled efficiently via dynamic programming, by solving shortest-path problems under hop constraints to generate feasible paths \parencite{ahuja1988network,Pioro}.\\
When the delay on a link is given by $1/(c_a - y_a)$, and the goal is to to minimize the average end-to-end delay of the demands over a network, the resulting objective function is then $
\sum_{a} {y_a}/{(c_a - y_a)}$,
which can be minimized within a convex optimization framework.
The survey~\parencite{ouorou2000survey} provides a comprehensive overview of solution approaches for convex-cost multi-commodity flow problems, in which the objective is to minimize the total cost—defined as the sum of convex arc costs—while satisfying all demand requirements. \\
The average end-to-end delay minimization problem has also been considered with the additional constraint that each flow should be unsplittable. In \parencite{yen2001}, a near-optimal solution is computed through a Lagrangian relaxation when delay constraints follow Kleinrock functions, while  \parencite{fortz2017} proposes a branch-and-bound approach to solve this problem when average delays are approximated by piecewise-linear functions. A  branch-and-price method is described in \parencite{beraudsudreau2026multi} to solve the problem for convex  delay functions.

The work of~\parencite{beker2003}, that formalizes delay constraints and provides some heuristic solutions, is one of the earliest and most closely related contributions to the present study. 
The cost version of the DCMCF problem was further investigated in~\parencite{ben2006mathematical}, where its weak NP-hardness was established through a reduction from the 2-Partition problem. In addition, convex relaxations were proposed to effectively handle the delay constraints. \\
A relaxation based on 
a disjunctive formulation of the delay constraints is proposed in \parencite{hijazi2012} ; more details will be provided in Section \ref{sec:hijazi}. 
 Furthermore, \parencite{hijazi2010outer} considers an outer-approximation of a big-M reformulation of the delay constraints. The same authors study in \parencite{hijazi2013robust}  a robust version of the problem, where a random component is added to the delay function.
While these works on the delay-constrained multi-commodity flow problem considered a set of paths given as input, in \parencite{truffot2010k-splittable}, the authors study a branch-and-price approach to solve a more general problem where the paths are built through an arc-node representation of the problem. In another approach, \parencite{duhamel2007augmented} and \parencite{papadimitriou2024augmented} provide efficient ways to compute locally optimal solutions of the delay-constrained multi-commodity flow problem.

Another form of delay constraints, where delays are proportional to link loads, has been considered in \parencite{bonami2017maximum}, where the authors prove the NP-completeness of the problem and provide polynomial approximations under assumptions on the paths and graph topologies.

\subsubsection*{On convex envelopes}

The convex relaxation of the delay constraints in Section \ref{sec:env} relates to the extensive literature on convex envelopes and, more generally, on convex under-estimator computations.
The best known and probably most widely used convex-estimator is obtained through the McCormick envelopes of the bilinear function  $xy$ on a rectangular domain \parencite{mccormick1976computability}. 
Several studies investigate the convex envelopes of function classes over special domains. We review here the works that are most closely related to the material presented in Section \ref{sec:env}.
 \parencite{sherali1990explicit} studies bilinear functions over ``D-polytopes" where a D-polytope is such that all edges have a non-negative slope. Notice that the function studied in Section \ref{sec:env} is not bilinear and the polytope considered there is not a D-polytope. 

More recently, \parencite{locatelli2018convex} studied the convex envelope of bivariate functions over polytopes under the assumption that the Hessian matrix is indefinite and that the restriction of the function along each edge is either concave or strictly convex. Their approach relies on solving a convex optimization problem. \\
Another approach proposed by the same author in \parencite{Loca16} is based on a polyhedral subdivision of the polytope and on the solution of a number of subproblems that grows exponentially with the number of vertices of the polytope. The special case of fractional functions is also investigated in \parencite{Loca16}. Since the function $x/(1-y)$ considered in Section \ref{sec:env} is fractional and satisfies all the properties mentioned above, the approach of \parencite{Loca16} could, in principle, also be used to compute its convex envelope. Nevertheless, we propose a more direct approach to derive the convex envelope, together with simple and self-contained proofs.

\subsection{Contributions and paper organization}

The remainder of the paper is organized as follows.

In Section~\ref{sect:complexity}, we analyze the computational complexity of the DCMCF problem and prove that it is strongly NP-hard, even when restricted to three commodities that share a common source and destination and are allowed to use all possible paths.
Section~\ref{sec:env} introduces a novel convex relaxation of the delay constraints. The resulting formulation is expressed through second-order cone constraints obtained from the convex envelope of the fractional function $x/(1-y)$.
We then show in Section~\ref{sec:alternative_relax} that our formulation outperforms the disjunctive programming approach proposed in~\parencite{hijazi2012}.
 Section~\ref{sec:low} is dedicated to theoretical analysis of heuristics for the DCMCF problem. We first show that heuristics based on convex over-estimators are essentially equivalent to explicitly selecting paths and solving a convex optimization problem defined over flow variables associated with those paths. We then propose several heuristics, including one derived from the new convex relaxation, and establish provable performance guarantees for it.
Finally, Section~\ref{sec:num_experiments} reports numerical experiments comparing the performance of the different relaxations and heuristics.

\section{Complexity \label{sect:complexity}}

We prove that the decision problem associated with DCMCF is strongly NP-hard, even when there are at most three commodities, when all commodities share the same source and the same sink, and when the available path sets $P^k$ include all possible paths for every commodity $k \in K$.

This provides a stronger characterization than existing results in the literature. Previous work by \parencite{ben2006mathematical} proved that the problem is weakly NP-hard, using a reduction from a 2-partition problem with a large number of commodities having distinct sources and destinations. Another form of delay constraints, where delays are proportional to link loads, was considered in \parencite{bonami2017maximum}. The authors proved the NP-completeness of that variant and provided polynomial approximations under specific assumptions on the paths and graph topologies.

\bigskip

We consider a reduction from the \textsc{3-Partition} problem \parencite{garey1975complexity}. This problem is defined as follows: decide whether a set $I$ of $n$ objects with sizes $\{w_i\}_{i \in I}$ can be partitioned into $m = n/3$ disjoint subsets (triplets) $\{T_j\}_{j \in J}$, such that each subset has the same sum $W/m$, where $W = \sum_{i \in I} w_i$.

To prove the NP-hardness of DCMCF, we build an instance of DCMCF that admits a maximum concurrent flow of value $\gamma \geq 1$ if and only if the given \textsc{3-Partition} instance is of \emph{True} type.

Since scaling all object weights by the same positive factor does not affect the \textsc{3-Partition} feasibility, we assume without loss of generality that the total weight is $W = \frac{1}{(2n)^2}$. We also assume that $w_i \leq W/m$ for each object $i \in I$, as the \textsc{3-Partition} instance is otherwise trivially infeasible, and that $m \geq 2$ (or $n \geq 6$) --- as the \textsc{3-Partition} problem involving a single triplet is trivially feasible. Finally, we denote $N = (2n)^2$. 

Given such a \textsc{3-Partition} instance, we construct a DCMCF instance on a directed graph $D=(V,A)$ as illustrated in Figure \ref{man:fig:dc-mcf:complexity}.

\begin{figure}[ht]
\centering
\includegraphics[page=5,trim={3cm 0.5cm 3cm 1.5cm},clip, width=0.8\textwidth]{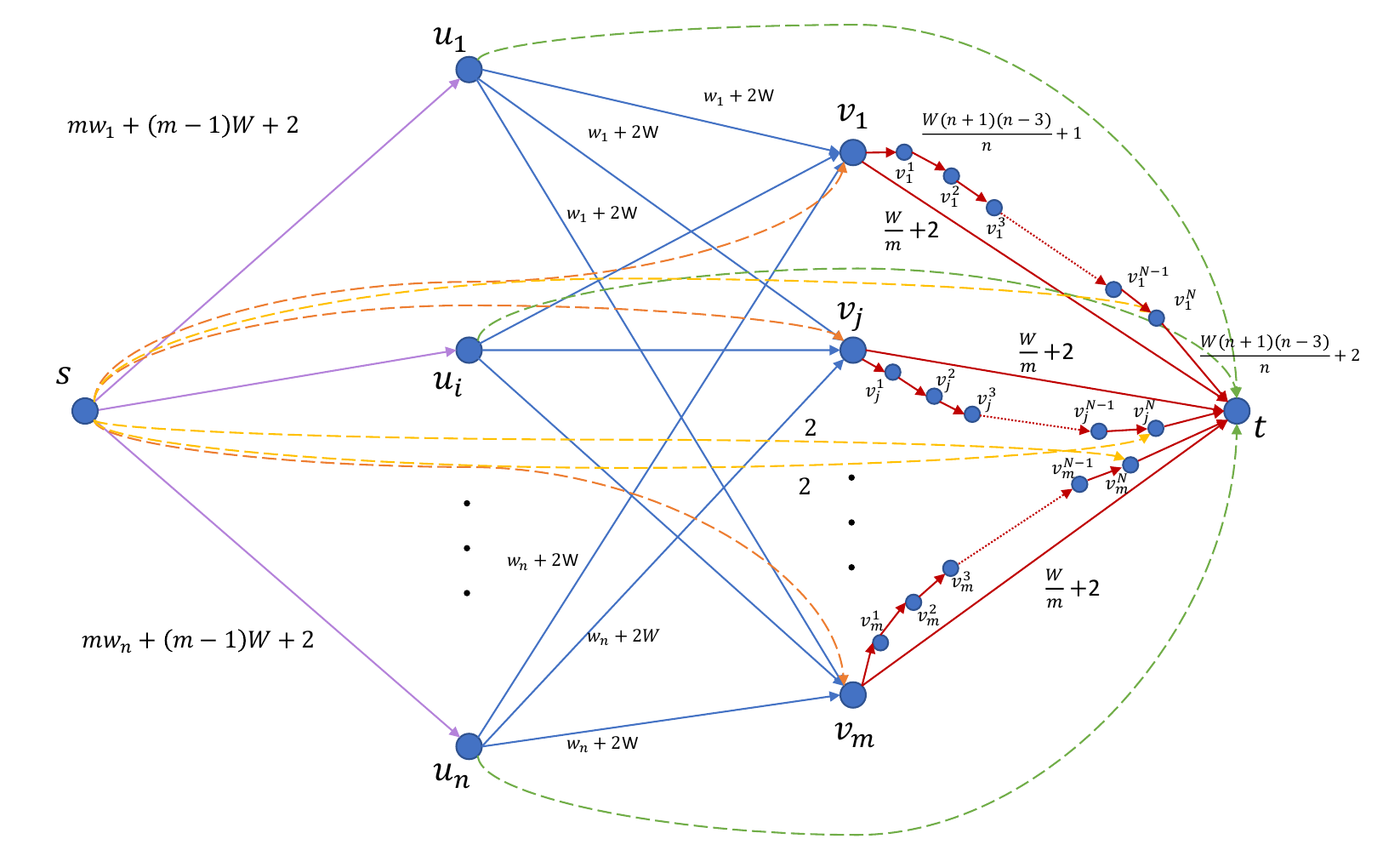}
\caption{Instance of the DCMCF problem associated with the \textsc{3-Partition} problem of a set $\{w_i\}_{i \in I}$. Only selected arcs, vertices, and capacities are shown. All dashed arcs have a capacity of 2.}
\label{man:fig:dc-mcf:complexity}
\end{figure}

\medskip

The set of vertices $V$ consists of the following elements:
\begin{itemize}
    \item A \textbf{source} vertex $s$ and a \textbf{destination} vertex $t$.
    \item \textbf{Object vertices:} a first layer of $n$ vertices $\{u_i\}_{i \in I}$, where each vertex corresponds to an object in the \textsc{3-Partition} instance.
    \item \textbf{Triplet vertices:} a second layer of $m = n/3$ vertices $\{v_j\}_{j \in J}$, where each vertex corresponds to a target triplet $j \in J$.
    \item \textbf{Chains of vertices:} for each $j \in J$, a sequence of vertices $\{v_j^l\}_{l \in \{1,\dots,N\}}$. We will prove that the subpath traversing these vertices is accessible to exactly one commodity, which will force its routing and flow.
\end{itemize}

\medskip

The set of arcs $A$ contains the following elements (colors refer to Figure \ref{man:fig:dc-mcf:complexity}):
\begin{itemize}
    \item \textbf{Source-objects arcs (purple):} $\forall i \in I$, an arc $(s,u_i)$ with capacity $c_{(s,u_i)} = m w_i + (m-1)W + 2$.
    \item \textbf{Partition arcs (blue):} $\forall i \in I, j \in J$, an arc $(u_i,v_j)$ with capacity $c_{(u_i,v_j)} = w_i + 2W$.
    \item \textbf{Triplet completion arcs (red):} $\forall j \in J$, an arc $(v_j,t)$ with capacity $c_{(v_j,t)} = \frac{W}{m} + 2$.
    \item \textbf{Chain arcs (red):} $\forall j \in J$, an entry arc $(v_j,v_j^1)$ and internal arcs $(v_j^l,v_j^{l+1})$ for $l \in \{1,\dots,N-1\}$, each with capacity $W \frac{(n+1)(n-3)}{n} + 1$. An exit arc $(v_j^N,t)$ is added with capacity $W \frac{(n+1)(n-3)}{n} + 2$.
    \item \textbf{Bypass arcs:} Three sets of arcs with a constant capacity of 2 are added to control delays: $(u_i, t)$ for all $i \in I$ (green dashed), $(s, v_j)$ for all $j \in J$ (orange dashed), and $(s, v_j^N)$ for all $j \in J$ (yellow dashed).
\end{itemize}

\medskip

We define three commodities sharing the same source $s$ and destination $t$. The structure and purpose of each commodity are summarized below:

\begin{itemize}
    \item \textbf{Commodity $k_L$ (Light and Fast):} Size $b_L = W$, maximum delay $d_L = 2 + \frac{1}{2W}$. 
    
    We will prove that it must be carried exclusively through paths of the form $s \rightarrow u_i \rightarrow v_j \rightarrow t$, effectively building the solution to the \textsc{3-Partition} problem.
    \item \textbf{Commodity $k_H$ (Heavy and Slow):} Size $b_H = (m-1)(n+1)W$, maximum delay $d_H = N + 2 + \frac{1}{W}$. 
    
    This commodity will use all paths crossing the partition arcs $(u_i, v_j)$ that are not used by $k_L$. This creates high delays on these arcs, forcing $k_L$ to define a partition of the objects set $I$.
    \item \textbf{Commodity $k_2$ (Delay calibration):} Size $b_2 = 5m$, maximum delay $d_2 = 2$. 
    
    We will show that this tight delay limit forces the commodity to use only 2-hop paths (formed by the bypass arcs) and to maintain a load of exactly 1 on each of these paths. This sets the delay and load on the crossed links.
\end{itemize}

\bigskip

\begin{proposition}
    If the \textsc{3-Partition} instance with weights $\{w_i\}_{i \in I}$ is feasible, then the related DCMCF instance described above admits a feasible solution with $\gamma = 1$. \label{prop:complex:part_implies_flow}
\end{proposition}

\begin{proof}
Assuming that the \textsc{3-Partition} instance is feasible, each vertex $v_j$ is assigned to a single triplet $T_j$. We can build a multi-commodity flow using the following paths:
\begin{itemize}
    \item For commodity $k_L$: for every triplet $j \in J$ and object $i \in T_j$, the path $s \rightarrow u_i \rightarrow v_j \rightarrow t$ is used with a load equal to $w_i$.
    \item For commodity $k_H$: for every triplet $j \in J$ and object $i \in I \setminus  T_j$, the path $s \rightarrow u_i \rightarrow v_j \rightarrow v_j^1 \rightarrow \dots \rightarrow v_j^N \rightarrow t$ is used with a load equal to $w_i + W$.
    \item For commodity $k_2$: each path in $\{s \rightarrow u_i \rightarrow t\}_{i \in I} \cup \{s \rightarrow v_j \rightarrow t\}_{j \in J} \cup \{s \rightarrow v_j^N \rightarrow t\}_{j \in J}$ is used with a load equal to $1$.
\end{itemize}

One can verify that each commodity sends the required load: $\sum_{i \in I} w_i = W = b_L$ for commodity $k_L$, $\sum_{j \in J}\sum_{i \in I \setminus T_j}(w_i+W)=m\left(W-\frac{W}{m}+(n-3)W\right)=mW(n-2-\frac{1}{m})=(m-1)(n+1)W=b_H$ for commodity $k_H$ and $2m+n=5m=b_2$ for commodity $k_2$. 

The delay of each arc is given by:
\begin{itemize}
    \item $d_{(u_i,v_j)} = \frac{1}{(w_i + 2W)-w_i}=\frac{1}{2W}$ if $i \in T_j$;
    \item $d_{(u_i,v_j)} = \frac{1}{(w_i + 2W)-(w_i+W)}=\frac{1}{W}$ if $i \notin T_j$;
    \item the delay of all other arcs is exactly equal to $1$.
\end{itemize}
It is easy to verify that the delay constraints are satisfied (with equality) for all paths used by all commodities.
\end{proof}

\bigskip

Let us now suppose that this $\text{DCMCF}$ instance admits a solution with $\gamma \geq 1$.
We will prove that the corresponding \textsc{3-Partition} instance is of True type. For this purpose, we will prove that commodities $k_2$, $k_L$ and $k_H$ are necessarily routed as described in the proof of Proposition \ref{prop:complex:part_implies_flow}, which allows us to infer the feasibility of the \textsc{3-Partition} instance. 

Since the proof is relatively long, we will only show the intermediate lemmas below: all proofs can be found in the Appendix section at the end of the paper.

\medskip 

Let us first prove that commodity $k_2$ must follow a specific set of paths:
\begin{lemma}
    Commodity $k_2$ must necessarily follow paths in
    $$ \{s \rightarrow u_i \rightarrow t\}_{i \in I} \cup \{s \rightarrow v_j \rightarrow t\}_{j \in J} \cup \{s \rightarrow v_j^N \rightarrow t\}_{j \in J}.$$
\label{lemma:complex:path_k2}
\end{lemma}

The proof of this lemma is presented in Appendix \ref{proof:complex:path_k2}. It relies on the fact that the minimum delay of a given arc $a \in A$ is $\frac{1}{c_a}$: we can show that, for any path $p$ from $s$ to $t$ not included in the set above, the minimum delay $\sum_{a \in p}\frac{1}{c_a}$ is greater than $d_2$.

\medskip

The next lemma states that the paths described in Lemma \ref{lemma:complex:path_k2} are actually used by $k_2$. 

\begin{lemma}
    $k_2$ uses all paths in $\{s \rightarrow u_i \rightarrow t\}_{i \in I} \cup \{s \rightarrow v_j \rightarrow t\}_{j \in J} \cup \{s \rightarrow v_j^N \rightarrow t\}_{j \in J}$.
    \label{lemma:complex:uses}
\end{lemma}

This lemma is proved by assuming that one of the paths is not used and showing that the additional load carried on the remaining paths would increase the path's delay above the commodity's limits. The full proof can be found in Appendix \ref{proof:complex:uses}.

\medskip

In order to compute the exact load of commodity $k_2$ over each used path, we define for each $i \in I$ and $j \in J$ the scalars $\epsilon_i$, $\epsilon_j$, and $\epsilon'_j$ in $]-1,1[$ such that: 
$$y_{(u_i,t)} = 1+\epsilon_i \quad  \forall i \in I, \qquad \qquad y_{(s,v_j)} = 1+\epsilon_{j} \quad \forall j \in J, \qquad \qquad  y_{(s,v_{j}^N)}  = 1+\epsilon'_{j} \quad \forall j \in J.$$
We also consider the sets $Q_I \subset I$, $Q_J \subset J$, and $Q_J' \subset J$ defined as follows:
$$ Q_I=\{i \in I \mid \epsilon_i \neq 0\}, \qquad Q_J=\{j \in J \mid \epsilon_j \neq 0\}, \qquad Q_J'=\{j \in J \mid \epsilon'_j \neq 0\}. $$
We will show in Lemma \ref{lemma:complex:load_k2} that $Q_I \cup Q_J \cup Q_J'$ is necessarily empty. In order to prove this result, we will first give upper bounds on the load of arcs $(s,u_i)$ for every $i \in I$ in Lemma \ref{lemma:complex:ineq_ysu}, and upper bounds on the load of arcs $(v_j,t)$ and $(v_j^N,t)$ for every  $j \in J$ in Lemma \ref{lemma:complex:ineq_yvt}.

\begin{lemma}
    For $i \in I$, the load of arc $(s,u_i)$ is strictly less than $m w_i + (m-1)W + 1 - \epsilon_i$ if $i \in Q_I$, and less than or equal to $m w_i + (m-1)W + 1$ if $i \in I \setminus Q_I$. \label{lemma:complex:ineq_ysu}
\end{lemma}

The proof of this lemma can be found in Appendix \ref{proof:complex:ineq_ysu}. It relies on the fact that commodity $k_2$ uses the arcs $(u_i,t)$ - enforcing a maximum delay of $b_2=2$ on the path $s \rightarrow u_i \rightarrow t$, and on a linear under-estimator of the delays of arcs $(s, u_i)$ and $(u_i,t)$.

\begin{lemma}
    For $j \in J$, the load of arc $(v_j,t)$ is strictly less than $\frac{W}{m}+1-\epsilon_j$ if $j \in Q_J$, and less than or equal to $\frac{W}{m}+1$ if $j \in J \setminus Q_J$.
    
    Similarly, the load of arc $(v_j^N,t)$ is strictly less than $W\frac{(n+1)(n-3)}{n}+1-\epsilon'_j$ if $j \in Q'_J$, and less than or equal to $W\frac{(n+1)(n-3)}{n}+1$ if $j \in J \setminus Q'_J$.\label{lemma:complex:ineq_yvt}
\end{lemma}

The proof of this lemma, given in Appendix \ref{proof:complex:ineq_yvt}, follows the same reasoning as the proof of Lemma \ref{lemma:complex:ineq_ysu}.

\bigskip

Based on Lemmas \ref{lemma:complex:ineq_ysu} and \ref{lemma:complex:ineq_yvt} above, we can establish the exact flow on several arcs of $A$:
\begin{lemma}
    The delays on arcs $(s,u_i)$, $(v_j,t)$,  $(v_{j}^N,t)$ and on the bypass arcs are equal to $1$ for every $i \in I$ and $j \in J$, and the loads on these arcs are:
    \begin{align}
        y_{(s,u_i)} &= m w_i + (m-1)W + 1 & \forall i \in I \label{eq:complex:load_sui} \\
        y_{(v_j,t)} &= \frac{W}{m} + 1 & \forall j \in J \label{eq:complex:load_vjt} \\
        y_{(v_{j}^N,t)} &= W\frac{(n+1)(n-3)}{n} + 1 & \forall j \in J \label{eq:complex:load_vjNt} \\
        y_{(u_i,t)} &= y_{(s,v_j)} = y_{(s,v_{j}^N)} = 1 & \forall j \in J, i \in I. \label{eq:complex:load_1}
    \end{align}
\label{lemma:complex:load_k2}
\end{lemma}

This lemma is proven by summing the loads on all arcs exiting $s$ and entering $t$, and showing that if the set of arcs $Q_I \cup Q_J \cup Q'_J$ is non-empty, this sum of loads is strictly lower than twice the size of the commodities $k_2$, $k_L$, and $k_H$, contradicting the possibility of a throughput $\gamma \geq 1$.

\bigskip
Equations \eqref{eq:complex:load_1} and Lemma \ref{lemma:complex:load_k2} together imply that commodity $k_2$ routes a unit load through each path in the set
$$ \{ s \rightarrow u_i \rightarrow t \}_{i \in I} \cup \{ s \rightarrow v_j \rightarrow t \}_{j \in J} \cup \{ s \rightarrow v_j^N \rightarrow t \}_{j \in J}. $$
As a consequence, the arcs $\{(u_i,t)\}_{i \in I}$, $\{(s,v_j)\}_{j \in J}$, and $\{(s,v_j^N)\}_{j \in J}$ are fully saturated by commodity $k_2$ if its delay is satisfied, meaning they cannot be used by commodity $k_L$ or $k_H$.

\bigskip

Let us now characterize the routing of commodity $k_L$.

\begin{lemma}
    Commodity $k_L$ must follow paths in $\{s \rightarrow u_i \rightarrow v_j \rightarrow t\}_{i \in I, j \in J}$.  \label{lemma:complex:path_kL}
\end{lemma}

The proof of this lemma, given in Appendix \ref{proof:complex:path_kL}, simply relies on the fact that all other paths are either saturated by commodity $k_2$ or violate the delay limit $d_L$.

\begin{lemma}
    The load of commodity $k_L$ over each arc in $\{(v_j,t)\}_{j \in J}$ is exactly $\frac{W}{m}$. \label{lemma:complex:load_kL}
\end{lemma}

The proof of this lemma can be found in Appendix \ref{proof:complex:load_kL} ; it directly results from the load of arc $(v_j,t)$ given in Equation \eqref{eq:complex:load_vjt}.

\bigskip

Let us now consider commodity~$k_H$.
\begin{lemma}
    The load of commodity $k_H$ over each arc in $\{(v_j, v_j^1)\}_{j \in J} \cup \{(v_j^l, v_j^{l+1})\}_{j \in J, l \in [1, N-1]} \cup \{(v_j^N, t)\}_{j \in J}$ is exactly $W\frac{(n+1)(n-3)}{n}$, and the delay on each arc is equal to $1$.     \label{lemma:complex:load_kH}
\end{lemma}

The proof of this lemma, which results from Equation \eqref{eq:complex:load_vjNt}, is given in Appendix \ref{proof:complex:load_kH}.

\bigskip

The results presented above fix the flow on each arc entering vertices $u_i$ for $i \in I$ and exiting vertices $v_j$ for $j \in J$ if the concurrent throughput $\gamma$ is greater than $1$.
We will now study the arcs $(u_i, v_j)$ to describe how the paths followed by commodity $k_L$ allow the construction of a \textsc{3-Partition} of the original problem. The next lemma  allows us to associate each object $i \in I$ with a unique part $j \in J$ in the original \textsc{3-Partition} problem:

\begin{lemma}
    For every $i \in I$, commodity $k_L$ is carried by a single arc exiting vertex $u_i$, and the load of $k_L$ on this arc is exactly $w_i$. This arc does not carry commodity $k_H$. \label{lemma:complex:single_arc_kL}
\end{lemma}

This lemma is proved in Appendix \ref{proof:complex:single_arc_kL}; this proof relies on the upper bound of the delays of arcs $(u_i,v_j)$ if commodities $k_L$ and $k_H$ cross them, and on the total flow exiting the nodes $u_i$ --- obtained by flow conservation and Equation \eqref{eq:complex:load_sui}.

\begin{lemma}
     For every $j \in J$, commodity $k_L$ is carried by exactly $3$ arcs entering vertex $v_j$. \label{lemma:complex:three_arcs_kL}
\end{lemma}

This result is proved by using the possible loads of arcs $(u_i,v_j)$, which depend on the commodities they carry, as shown in Lemma \ref{lemma:complex:single_arc_kL}, and the flow conservation at node $v_j$.

\bigskip

\begin{proposition}
    If the  DCMCF instance described above has a feasible solution of value $\gamma =1$, the \textsc{3-Partition} instance with weights $\{w_i\}_{i \in I}$ is of type True. \label{prop:complex:flow_implies_part}
\end{proposition}

Lemmas \ref{lemma:complex:single_arc_kL} and \ref{lemma:complex:three_arcs_kL} showed that if the $\text{DCMCF}$ instance is feasible, the arcs $(u_i,v_j)$ crossed by commodity $k_L$ allow the construction of triplets of objects (and that each arc carries a load $w_i$) ; flow conservation and Equation \eqref{eq:complex:load_vjt} allow us to conclude that the sum of weights of each triplet is $\frac{W}{m}$, showing that the \textsc{3-Partition} instance was feasible.
The complete proof is given in Appendix \ref{proof:complex:flow_implies_part}.

\begin{theorem}
 DCMCF is strongly NP-hard, even when there are at most $3$ commodities, all sharing the same source and destination, and when every path from the source to the destination can potentially be used by each commodity.
 \label{theorem:complex}
\end{theorem}

\begin{proof}
This is a direct consequence of  Propositions \ref{prop:complex:part_implies_flow} and \ref{prop:complex:flow_implies_part} and the strong NP-hardness of \textsc{3-Partition}.
\end{proof}

\section{Relaxation of the delay constraints}
\label{sec:env}

Constraints \eqref{eq:dcmf:delay} rewritten as \eqref{eq:dcmf:delay_time_xkp} are the ``difficult'' constraints of the DCMCF problem. In fact, the function $g_{p}: \{x^k_p, (y_a)_{a \in p}\}  \rightarrow   \sum_{a \in p} \frac{x^k_p}{c_a - y_a} - x^k_p d^k$ is a non-convex function.  
One way to get good convex relaxations of DCMCF consists of replacing, for every $k \in K$ and $p \in P^k$, constraint $g_p(\{x^k_p, (y_a)_{a \in p}\}) \leq 0$ by constraint $h_p(\{x^k_p, (y_a)_{a \in p}\}) \leq 0$ where $h_p$  is convex and satisfies $h_p \le g_p$. Since $g_p$ is an additive function with a delay term for each arc of $p$, it is natural to decompose the problem by considering all terms separately.  This leads to the study of the function $f$ defined by  $f(x,y) = \frac{x}{1-y}$ for $(x,y) \in [0,1]^2$. 
Constraints \eqref{eq:dcmf:delay_time_xkp} can  then be formulated as: 
\begin{small_equation}
    \sum_{a \in p} f(\frac{b^k x^k_p}{c_a}, \frac{y_a}{c_a} ) \leq b^k x^k_p \times d^k  \qquad \forall k \in K, p \in P^k. \label{eq:constr_f}
\end{small_equation}
To obtain  good convex relaxations, one can then focus on the convex envelope of $f$.  More precisely, one can notice that, for every $k \in K$, $p \in P^k$ and $a \in A$, the domain of $\frac{b^k x_p^k}{c_a}$ and $\frac{y_a}{c_a}$ is narrower than $[0,1]^2$:
\begin{itemize}
    \item[-] Inequality \eqref{eq:dcmf:xa} ensures that $\frac{b^k x_p^k}{c_a} \leq \frac{y_a}{c_a}$.
    \item[-] Inequality \eqref{eq:dcmf:delay} ensures that $\frac{y_a}{c_a} < 1$ (otherwise the delay of link $a$ would be infinite, violating the delay constraint).
    \item[-] Finally, both the $x_p^k$ and $y_a$ variable are nonnegative.
\end{itemize}
Therefore we will compute the convex envelope $\check{f}$  of $f$ on the polyhedral subset $\mathcal{X}_{\epsilon, \theta, \zeta, \eta}$ defined for $(\epsilon, \zeta, \eta, \theta)$ in $]0,1[^2 \times [0,1]^2$ by: 
\begin{small_equation}
    \mathcal{X}_{\epsilon, \theta, \zeta, \eta} = \{(x,y) \in [0,1]^2: y \geq x+\theta, y \in [\eta, 1-\epsilon], x \leq 1- \zeta\}.
\end{small_equation}
As will be shown in Section \ref{sub:param}, given some $p \in P^k$ and $a \in p$, inequalities relating $x^k_p$ and $y_a$   can be easily used to deduce possible values of $(\epsilon, \theta, \zeta, \eta) \in [0,1]^4$. Note that these values will depend on $p$ and $a$ and one can solve some intermediate convex problems to get the best values. Before providing complete details in Section \ref{sub:param}, let us start for now by computing the convex envelope of $f$ on $\mathcal{X}_{\epsilon, \theta, \zeta, \eta}$.

%with $\epsilon, \zeta, \eta, \theta$ in $]0,1]^2 \times [0,1]^2$.

Let us first divide $\mathcal{X}_{\epsilon, \theta, \zeta, \eta}$ into $3$  polyhedral subsets (see Figure \ref{fig:domain} for illustration)
$\mathcal{X}_{\epsilon,\theta, \zeta, \eta} = \mathcal{R}_{\epsilon, \theta, \eta } \cup \mathcal{S}_{\epsilon, \theta, \zeta, \eta} \cup \mathcal{T}_{\epsilon, \theta,  \zeta}$, where  
{
\begin{small_align}
    \mathcal{R}_{\epsilon, \theta, \eta } = \bigg\{(x,y) \in [0,1]^2:& y \geq \eta, y \leq 1-\epsilon - \frac{1-\epsilon-\eta}{\eta-\theta} x\bigg\}  \nonumber \\
    \mathcal{S}_{\epsilon,\theta,  \zeta, \eta} = \bigg\{(x,y) \in [0,1]^2:& y \geq x+\theta, 
      y \geq 1-\epsilon - \frac{1-\epsilon-\eta}{\eta-\theta} x,  y \leq 1-\epsilon - \frac{\zeta-\epsilon-\theta}{1-\zeta} x\bigg\}  \nonumber \\    \mathcal{T}_{\epsilon,\theta,  \zeta} = \bigg\{(x,y) \in [0,1]^2:& y \leq 1-\epsilon, x \leq 1-\zeta, y \geq 1-\epsilon - \frac{\zeta-\epsilon - \theta}{1-\zeta} x  \bigg\}. \nonumber
\end{small_align}}
Notice that the partition makes sense (i.e., the $3$ subsets are well-defined) only if  
\begin{small_equation}
    0 \leq \theta \leq \eta \leq 1-\zeta + \theta \leq 1-\epsilon \leq 1 .
    \label{eq:condi}
\end{small_equation} 

We will compute the convex envelope of $f$ on each of the subsets $\mathcal{R}_{\epsilon, \theta, \eta }$, $\mathcal{S}_{\epsilon, \theta, \zeta, \eta}$ and $\mathcal{T}_{\epsilon, \zeta, \eta}$. Then we will show that they can be combined to obtain the convex envelope on $\mathcal{X}_{\epsilon, \theta, \zeta, \eta}$.

\begin{figure}[ht]
\centering
\includegraphics[page=2,trim={4.5cm 1cm 13cm 0.5cm},clip, width=.4\textwidth]{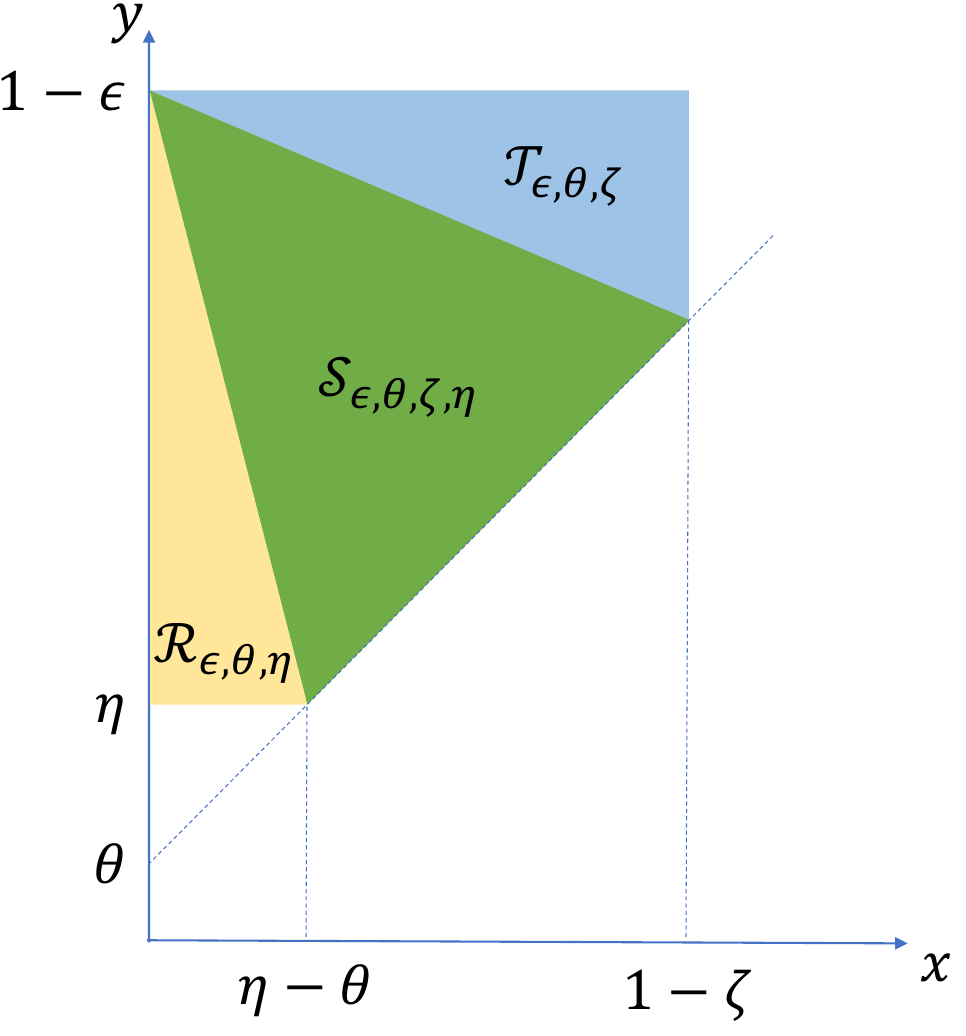}
\caption{
    Representation of the domain $\mathcal{X}_{\epsilon, \theta, \zeta, \eta}$ of the function $f$, and its partition in $\mathcal{R}_{\epsilon, \theta, \eta } \cup \mathcal{S}_{\epsilon, \theta, \zeta, \eta} \cup \mathcal{T}_{\epsilon, \theta,  \zeta}$.}
    %\vspace{-1cm}
     \label{fig:domain}
\end{figure}

\subsection{Convex envelope of $f$ on $\mathcal{R}_{\epsilon, \theta, \eta}$}

Consider the function $r_\eta: \mathcal{X}_{\epsilon, \theta, \zeta, \eta} \rightarrow \mathbb R$ where $r_\eta(x,y)= \frac{x}{1-\eta}$.

\begin{proposition}
$r_\eta$ is a convex under-estimator of $f$ on  $\mathcal{X}_{\epsilon, \theta, \zeta, \eta}$ (i.e., $\forall (x,y) \in  \mathcal{X}_{\epsilon, \theta, \zeta, \eta}$, \\ $r_\eta(x,y) \leq   \check{f}_{\mathcal{X}_{\epsilon, \theta, \zeta, \eta}}(x,y) 
\leq  f(x,y)$).  When $f$ is restricted to $\mathcal{R}_{\epsilon, \theta,\eta}$, its convex envelope is given by $r_\eta$  (i.e., $\check{f}_{\mathcal{R}_{\epsilon, \theta, \eta}}  = r_\eta$).
\label{pro:eta}
\end{proposition}
\begin{proof}
$r_\eta$ is linear and thus convex. We also have 
 $r_\eta \leq f$ on $\mathcal{X}_{\epsilon, \theta, \zeta, \eta}$ since $y \geq \eta$. \\
So $r_\eta$ is a convex under-estimator of $f$. To show that it is the convex envelope of $f$ on $\mathcal{R}_{\epsilon, \theta, \eta}$, let us show that $\check{f}_{\mathcal{R}_{\epsilon, \theta, \eta}}(x,y)  \leq r_\eta(x,y),  \forall (x,y) \in \mathcal{R}_{\epsilon, \theta, \eta}$.
Observe that $(x,y) \in  \mathcal{R}_{\epsilon, \theta, \eta }$ can be written  as 
 $\alpha (x_0, \eta) + (1-\alpha)(0,1-\epsilon)$ with $\alpha = \frac{x}{x_0} \in [0,1]$ and $x_0 = x \frac{1 - \epsilon - \eta}{1 - \epsilon  - y}$. Note that  $(x_0, \eta)$ and  $(0,1-\epsilon)$  both belong to $\mathcal{R}_{\epsilon, \theta, \eta }$, and that $r_\eta(x,y)=\alpha f(x_0, \eta) + (1-\alpha) f(0, 1-\epsilon)$. By convexity of $\check{f}_{\mathcal{R}_{\epsilon, \theta, \eta}}$, we must have $r_\eta \geq \check{f}_{\mathcal{R}_{\epsilon, \theta, \eta}}$, and thus $r_\eta$ is the convex envelope of $f$ restricted to $R_{\epsilon, \theta, \eta}$. 
\end{proof}

\subsection{Convex envelope of $f$ on $\mathcal{S}_{\epsilon, \theta, \eta, \zeta}$} \label{sec:conv_env:s}

We define the function $s_{\epsilon, \theta}: \mathcal{X}_{\epsilon, \theta, \zeta, \eta} \rightarrow \mathbb R$ as $s_{\epsilon,\theta}(x,y)= \frac{x}{1-\theta - \frac{x}{1- \frac{y-x-\theta}{1-\epsilon-\theta}}}$  if $(x,y) \neq (0, 1-\epsilon)$, and $s_{\epsilon, \theta}(0, 1-\epsilon)=0$.

\begin{proposition}
$s_{\epsilon, \theta}$ is a convex under-estimator of $f$ on  $\mathcal{X}_{\epsilon, \theta, \zeta, \eta}$.  When $f$ is restricted to $\mathcal{S}_{\epsilon, \theta, \zeta, \eta}$, its convex envelope is given by $s_{\epsilon, \theta}$  (i.e., $\check{f}_{\mathcal{S}_{\epsilon, \theta, \zeta, \eta}}  = s_{\epsilon, \theta}$).
\label{theorem:f_theta}
\end{proposition}

\begin{proof}

Let us first check that the denominator of  $s_{\epsilon, \theta}(x,y)$ is strictly positive. Observe that $1 - \theta  -\frac{x}{1- \frac{y-x-\theta}{1-\epsilon-\theta}} = 1 - \theta - x \frac{1-\epsilon-\theta}{1-\epsilon-y+x}$. $1-\epsilon-y+x$ is non-zero on $\mathcal{X}_{\epsilon, \theta, \zeta, \eta}$ if $(x,y) \neq (0, 1-\epsilon)$. Using $y-x \geq \theta$, we get $1 -\theta-\frac{x}{1- \frac{y-x-\theta}{1-\epsilon-\theta}} \geq 1-\theta-x  \geq 1-y \geq \epsilon > 0$. $s_{\epsilon, \theta}$ is then well-defined on $\mathcal{X}_{\epsilon, \theta, \zeta, \eta}$ and positive. It is also easy to prove the continuity of the function at $(0,1-\epsilon)$ since $\lim_{(x,y) \rightarrow (0,1-\epsilon)} s_{\epsilon, \theta}(x,y) = 0 = s_{\epsilon, \theta}(0, 1-\epsilon)$.

To prove the convexity of $s_{\epsilon, \theta}$ on $\mathcal{X}_{\epsilon, \theta, \zeta, \eta}$, one can for example use the fact that  the perspective function of a convex function is also convex.  
Consider the function $f_\theta$ defined on $[0,1[$ by $f_\theta (x)=f(x,x + \theta) = \frac{x}{1-x-\theta}$. $f_\theta$ is obviously convex. Therefore, its perspective function $p$  defined 
as $p_s(x,t)=t f_\theta(\nicefrac{x}{t})=\frac{x}{1-\theta-\nicefrac{x}{t}}$ is convex.
Since $s_{\epsilon, \theta}(x,y)=p_s(x, 1- \frac{y-x-\theta}{1-\epsilon-\theta})$ is the composition of a convex and an affine function, $s_{\epsilon,\theta}$ is convex.

To show that $s_{\epsilon, \theta}$ is an under-estimator of $f$ over $\mathcal{X}_{\epsilon, \theta, \zeta, \eta}$ for $x>0$, we should prove that, for $(x,y) \in \mathcal{X}_{\epsilon, \theta, \zeta, \eta}$,  $
\frac{s_{\epsilon, \theta}(x,y)}{f(x,y)} = \frac{1-y}{1- \theta-\frac{x}{1 -\frac{y-x-\theta}{1-\epsilon-\theta}}}  = \frac{(1-y)(1 - \epsilon + x - y)}{(1-\theta) (1 - \epsilon + x -y) - x(1-\epsilon-\theta)}$ is less than $1$. Considering the difference between the numerator and the denominator we get: 
\begin{small_equation*}
(1-y)(1 - \epsilon + x - y) - (1-\theta) (1 - \epsilon + x -y) + x(1-\epsilon-\theta) =  (\theta + x - y) (1 - \epsilon -y)     
\end{small_equation*}
which is obviously non-positive on the set $\mathcal{X}_{\epsilon, \theta, \zeta, \eta}$. For $x=0$, both functions are equal to $0$.

Finally, let us restrict $f$ to $\mathcal{S}_{\epsilon, \theta, \zeta, \eta}$. Since $s_{\epsilon, \theta}$ is a convex under-estimator of $f$, we already have $\check{f}_{\mathcal{S}_{\epsilon, \theta, \zeta, \eta}}  \geq s_{\epsilon, \theta}$. For the other direction, given any $(x,y) \in \mathcal{S}_{\epsilon, \theta, \zeta, \eta}$, let $\alpha=\frac{1-\epsilon + x- y }{1-\epsilon - \theta}$ and $x_0=y_0-\theta= \frac{x}{\alpha}$. Observe that  $(x_0,y_0)$ and $(0,1-\epsilon)$ are both in $\mathcal{S}_{\epsilon, \theta, \zeta, \eta}$,  and $(x,y)$ is a convex combination of $(x_0,y_0)$ and $(0,1-\epsilon)$: $(x,y)=\alpha (x_0,y_0) + (1-\alpha) (0,1-\epsilon)$.  Moreover, one can easily check that $s_{\epsilon, \theta}(x,y)= \alpha f(x_0,y_0) + (1-\alpha) f(0,1-\epsilon)$. By convexity of $\check{f}_{\mathcal{S}_{\epsilon, \theta, \zeta, \eta}}$, we deduce that $\check{f}_{\mathcal{S}_{\epsilon, \theta, \zeta, \eta}}  \leq s_{\epsilon, \theta}$.  
\end{proof}

The next proposition states that a constraint of form $s_{\epsilon, \theta}(x,y) \leq d$ can be expressed through second-order cone programming. 
\begin{proposition}
$s_{\epsilon, \theta}(x,y) \leq d$  can be expressed as the hyperbolic constraint $(z+d)( z- \frac{1}{1-\theta} x)  \geq  z^2$ and the linear constraint $z = 1-\frac{y-x - \theta}{1-\epsilon - \theta}$. 
\label{pro:socp_s}
\end{proposition}
\begin{proof}
Let $z = 1-\frac{y-x-\theta}{1-\epsilon-\theta}$. Then: 
\begin{small_equation}
    s_{\epsilon,\theta}(x,y)\leq d \quad \iff \quad \frac{x}{1 - \theta -\frac{x}{z}} \leq d \quad \iff \quad\frac{\frac{1}{1-\theta} x z}{ z-\frac{1}{1-\theta}x} + z \leq z + d  \quad \iff \quad \frac{ z^2 }{ z - \frac{1}{1-\theta} x} \leq  z + d.  
     \label{eq:socp_s}
\end{small_equation}
Using the fact that  $( z-\frac{1}{1-\theta} x )$ is strictly positive (from the  first part of the proof of Proposition \ref{theorem:f_theta}), the last inequality is equivalent to 
  $z^2 \leq (z + d) ( z- \frac{1}{1-\theta}x)$. 
\end{proof}
Note that Proposition \ref{pro:socp_s}  provides a second proof of the convexity of $s_{\epsilon, \theta}$ since its epigraph is shown to be convex. 

\subsection{Convex envelope of $f$ on $\mathcal{T}_{\epsilon, \theta, \zeta}$} \label{sec:conv_env:t}

Consider the function $t_{\epsilon, \zeta}: \mathcal{X}_{\epsilon, \theta, \zeta, \eta} \rightarrow \mathbb R$ where $t_{\epsilon, \zeta}(x,y)= \frac{x^2}{(1-\zeta) (1-\epsilon)+\epsilon x - (1-\zeta) y}$ if $(x,y) \neq (0, 1-\epsilon)$, and $t_{\epsilon, \zeta}(0, 1-\epsilon)=0$.

\begin{proposition}
$t_{\epsilon, \zeta}$ is a convex under-estimator of $f$ on  $\mathcal{X}_{\epsilon,\theta,\zeta, \eta}$.  When $f$ is restricted to $\mathcal{T}_{\epsilon, \theta, \zeta}$, its convex envelope is given by $t_{\epsilon, \zeta}$  (i.e., $\check{f}_{\mathcal{S}_{\epsilon, \theta, \zeta, \eta}}  =t_{\epsilon, \zeta}$).
\label{theorem:f_zeta}
\end{proposition}

\begin{proof}
The proof is similar to the proof of Proposition \ref{theorem:f_theta}.

First, the denominator in the fraction that defines $t_{\epsilon, \zeta}$ is strictly positive. Indeed, $(1-\zeta) (1-\epsilon) +\epsilon x - (1-\zeta) y = (1-\zeta)(1 - \epsilon  -y) + \epsilon x>0$ if $(x,y) \in \mathcal{X}_{\epsilon,\theta,\zeta, \eta} \setminus \{(0, 1-\epsilon)\}$. So $t_{\epsilon, \zeta}$ is well-defined on $\mathcal{X}_{\epsilon, \theta, \zeta, \eta}$. 
It is also easy to prove the continuity of the function at $(0,1-\epsilon)$ since $\lim_{(x,y) \rightarrow (0,1-\epsilon)} t_{\epsilon, \zeta}(x,y) = 0 = t_{\epsilon, \zeta}(0, 1-\epsilon)$.

Convexity of $t_{\epsilon, \zeta}$ can be proven by considering the convex function $f_{\zeta}$ defined on $[-\infty, 1-\epsilon]$ by $f_{\zeta}(y)=f(1-\zeta,y)=\frac{1-\zeta}{1-y}$. Because $f_\zeta$ is convex, its perspective function $p_t$ defined as $p_t(x,t)=t f_\zeta(\frac{x}{t})=\frac{t^2(1-\zeta)}{t-y}$ is convex.
Since $t_{\epsilon, \zeta}(x,y)=p_t(y-(1-\frac{x}{1-\zeta})(1-\epsilon), \frac{x}{1-\zeta})$ is the composition of a convex and an affine function, $t_{\epsilon, \zeta}$ is convex.

To show that $t_{\epsilon, \zeta}$ is an under-estimator of $f$ over $\mathcal{X}_{\epsilon, \theta, \zeta, \eta}$ for $x>0$, we should prove that, for $(x,y) \in \mathcal{X}_{\epsilon, \theta, \zeta, \eta}$,  $\frac{t_{\epsilon, \zeta}(x,y)}{f(x,y)} = \frac{x (1-y )}{(1-\zeta) (1-\epsilon)+\epsilon x - (1-\zeta) y}$ is less than $1$. Considering the difference between the numerator and the denominator, we get: 
\begin{small_equation*}
x - y x - (1-\zeta)(1-\epsilon) - \epsilon x + (1-\zeta) y = (1-\epsilon) (1-\zeta-x) + y (1-\zeta-x) = -(1-\epsilon - y)(1-\zeta-x)    
\end{small_equation*}
which is obviously nonpositive on $\mathcal{X}_{\epsilon, \theta \zeta, \eta}$ (for $x=0$, we have $f(x,y)=t_{\epsilon, \zeta}(x,y)=0$).

Finally, let us consider the restriction of  $f$  to $ \mathcal{T}_{\epsilon, \theta, \zeta}$. As $t_{\epsilon, \zeta}$ is a convex under-estimator of $f$, it satisfies $t_{\epsilon, \zeta}(x,y) \leq \check{f}_{\mathcal{T}_{\epsilon, \zeta}}(x,y)$. Observe that  each  $(x,y)$ of $\mathcal{T}_{\epsilon, \theta, \zeta}$ can be expressed as a convex combination of two members of  $\mathcal{T}_{\epsilon, \theta, \zeta}$:  $(x,y) = \alpha (1-\zeta, y_0) + (1-\alpha)(0,1-\epsilon)$, with $\alpha=\frac{x}{1-\zeta}$ and $y_0 = y\frac{1-\zeta}{x}-\frac{(1-\zeta)(1-\epsilon)}{x}+1-\epsilon$. Moreover, one can easily check that $t_{\epsilon, \zeta}(x,y) =  \alpha f(1-\zeta, y_0) + (1-\alpha) f(0, 1-\epsilon)$ proving, by convexity of $\check{f}_{\mathcal{T}_{\epsilon, \theta, \zeta}}$, that $t_{\epsilon, \zeta}(x,y) \geq \check{f}_{\mathcal{T}_{\epsilon, \theta, \zeta}}(x,y)$.   
\end{proof}

\begin{remark}
The constraint $t_{\epsilon, \zeta}(x,y) \leq d$ can be expressed through second-order cone programming since the  inequality $d  \times \left((1-\zeta) (1-\epsilon)+\epsilon x - (1-\zeta) y\right) \geq x^2$ is a convex hyperbolic constraint.   \label{rem:socp_t}
\end{remark}
Note that Remark \ref{rem:socp_t} also allows to prove convexity of $ t_{\epsilon, \zeta}$.  

\subsection{Convex envelope of $f$ on $\mathcal{X}_{\epsilon, \theta, \zeta, \eta}$}

\begin{theorem}
The convex envelope of $f$ on $\mathcal{X}_{\epsilon, \theta, \zeta, \eta}$ is given by: 
{%\normalsize
\begin{small_align}
    \check f_{\mathcal{X}_{\epsilon, \theta, \zeta, \eta}} (x,y)
    = & \max(r_{\eta}(x,y), s_{\epsilon, \theta}(x,y), t_{\epsilon, \zeta}(x,y)) \nonumber \\
   \check f_{\mathcal{X}_{\epsilon, \theta, \zeta, \eta}} (x,y)
    = & \max(\frac{x}{1-\eta}, \frac{x}{1-\theta - \frac{x}{1- \frac{y-x-\theta}{1-\epsilon-\theta}}}, \frac{x^2}{(1-\zeta) (1-\epsilon)+\epsilon x - (1-\zeta) y}).
\end{small_align}}
\label{theorem:cvx_env}
\end{theorem}
\begin{proof}
From Propositions \ref{pro:eta}, \ref{theorem:f_theta} and \ref{theorem:f_zeta}, $\max(r_{\eta}(x,y), s_{\epsilon, \theta}(x,y), t_{\epsilon, \zeta}(x,y))$ is a convex under-estimator of $f$ on $\mathcal{X}_{\epsilon, \theta, \zeta, \eta}$, showing that $\check f_{\mathcal{X}_{\epsilon, \theta, \zeta, \eta}} \geq \max(r_{\eta}(x,y), s_{\epsilon, \theta}(x,y), t_{\epsilon, \zeta}(x,y))$.

To show equality, first observe that $\mathcal{R}_{\epsilon, \theta, \eta } \subset \mathcal{X}_{\epsilon, \theta, \zeta, \eta}$ implies that $\check f_{\mathcal{X}_{\epsilon, \theta, \zeta, \eta}}(x,y) \leq \check f_{\mathcal{R}_{\epsilon, \theta, \eta}}(x,y)$ for any $(x,y) \in  \mathcal{R}_{\epsilon, \theta, \eta }$. 

Similarly, $\check f_{\mathcal{X}_{\epsilon, \theta, \zeta, \eta}}(x,y) \leq \check f_{\mathcal{S}_{\epsilon, \theta, \eta, \zeta}}(x,y)$ for any $(x,y) \in  \mathcal{S}_{\epsilon, \theta, \eta, \zeta}$, and 
$\check f_{\mathcal{X}_{\epsilon, \theta, \zeta, \eta}}(x,y) \leq \check f_{\mathcal{T}_{\epsilon, \theta, \zeta}}(x,y)$ for any $(x,y) \in  \mathcal{T}_{\epsilon, \theta, \zeta}$. 

Using that $\check f_{\mathcal{R}_{\epsilon, \theta, \eta}} =r_{\eta}$, $\check f_{\mathcal{S}_{\epsilon, \theta, \eta, \zeta}} = s_{\epsilon, \theta} $ and 
$ \check f_{\mathcal{T}_{\epsilon, \theta, \zeta}}=t_{\epsilon, \zeta} $, we deduce that: 
\begin{small_align*}
& \forall (x,y) \in \mathcal{R}_{\epsilon, \theta, \eta }, & \check f_{\mathcal{X}_{\epsilon, \theta, \zeta, \eta}}(x,y) \leq  r_\eta(x,y) = \max(r_{\eta}(x,y), s_{\epsilon, \theta}(x,y), t_{\epsilon, \zeta}(x,y)),  \\
&\forall (x,y) \in \mathcal{S}_{\epsilon, \theta \eta, \zeta}, &\check f_{\mathcal{X}_{\epsilon, \theta, \zeta, \eta}}(x,y) \leq  s_{\epsilon, \theta}(x,y) =\max(r_{\eta}(x,y), s_{\epsilon, \theta}(x,y), t_{\epsilon, \zeta}(x,y)), \\
&\forall (x,y) \in \mathcal{T}_{\epsilon, \theta, \zeta}, &\check f_{\mathcal{X}_{\epsilon, \theta, \zeta, \eta}}(x,y) \leq  t_{\epsilon, \zeta}(x,y) = \max(r_{\eta}(x,y), s_{\epsilon, \theta}(x,y), t_{\epsilon, \zeta}(x,y)).  
\end{small_align*}
Since $\mathcal{X}_{\epsilon, \theta, \zeta, \eta} = \mathcal{R}_{\epsilon, \theta, \eta } \cup \mathcal{S}_{\epsilon, \theta, \eta, \zeta} \cup \mathcal{T}_{\epsilon, \theta, \eta}$,
$
\check f_{\mathcal{X}_{\epsilon, \theta, \zeta, \eta}}(x,y) \leq  \max(r_{\eta}(x,y), s_{\epsilon, \theta}(x,y), t_{\epsilon, \zeta}(x,y))$   
for any $(x,y) \in \mathcal{X}_{\epsilon, \theta, \zeta, \eta}$.
\end{proof}

Figure~\ref{fig:conv_env_comparisons} illustrates the difference between the values of the function $f$ and those of its under-estimators, highlighting the tightness of the relaxation.

Observe that the relaxation is tighter when $x$ is large or when $y - x$ is small, corresponding respectively to high path usage and low arc usage by the other paths in the DCMCF problem.

\begin{figure}[ht]
    \includegraphics[trim={0cm 0cm 0cm 1.cm},clip, width=1\linewidth]{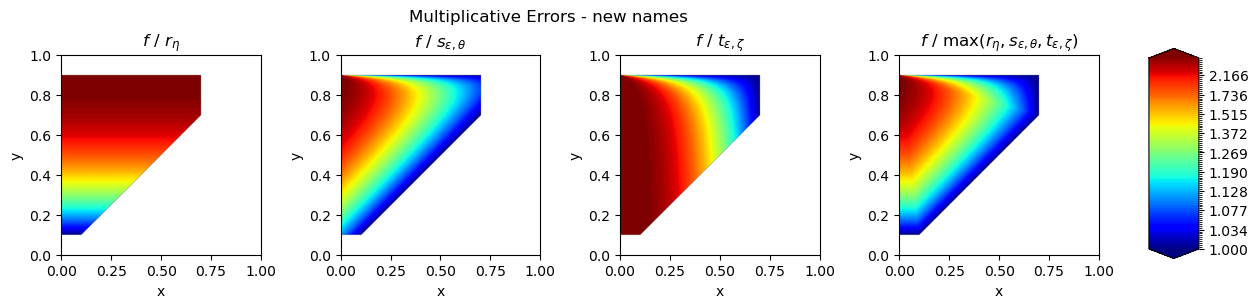}
    \caption{Values and multiplicative errors of $r_\eta$, $s_{\epsilon, \theta}$, $t_{\epsilon, \theta}$ and $\check f_{\mathcal{X}_{\epsilon, \theta, \zeta, \eta}}$ with $\epsilon=0.1$, $\zeta=0.3$, $\eta=0.1$, and $\theta=0$. Values are lower (in blue) when relaxations are closer to the actual value of $f$.}
    \label{fig:conv_env_comparisons}
\end{figure}

\begin{proposition}
A constraint of the type $\check f_{\mathcal{X}_{\epsilon, \theta, \zeta, \eta}}(x,y) \leq d, \forall (x,y) \in \mathcal{X}_{\epsilon, \theta, \zeta, \eta}$ can be formulated as a set of second-order cone programming constraints. \label{pro:socp}
\end{proposition}

\begin{proof}
    A constraint of type $\check f_{\mathcal{X}_{\epsilon, \theta, \zeta, \eta}}(x,y) = \max(r_{\eta}(x,y), s_{\epsilon, \theta}(x,y), t_{\epsilon, \zeta}(x,y))  \leq d$ can be expressed as a set of 3 constraints $r_{\eta}(x,y) \leq d$, $s_{\epsilon, \theta}(x,y) \leq d$ and $t_{\epsilon, \zeta}(x,y)) \leq d$. From Proposition \ref{pro:socp_s} and Remark \ref{rem:socp_t}, the last two constraints can be formulated as second-order cone constraints, and the first one is linear.
\end{proof}

\subsection{Relaxation of the DCMCF problem} \label{subsec:env_relaxation}

Constraints \eqref{eq:constr_f} provide a reformulation of the delay constraints \eqref{eq:dcmf:delay} of the DCMCF problem based on the function $f$. From this reformulation and the convex envelope of the function $f$, the delay constraints of the DCMCF problem can be relaxed as:
\begin{small_equation}
    \sum_{a \in p} \check f_{\mathcal{X}_{\frac{c_a-u_a}{c_a}, \frac{{l_a}^k_p}{c_a}, \frac{c_a - b^k u^k_p}{c_a}, \frac{l_a}{c_a}}} (\frac{b^k}{c_a} {x^k_p},\frac{y_a}{c_a}) \leq d^k b^k x^k_p \qquad \forall k \in K, p \in P^k \label{eq:constr_f_check}
\end{small_equation}
where $u_a$ (resp. $u^k_p$) is an upper bound of $y_a$ (resp. $x^k_p$), and $l_a$  is (resp. ${l_a}^k_p$) a lower bound on $y_a$ (resp. $y_a - b^k x^k_p$). Details about the computation of these parameters are deferred to Section \ref{sub:param}. 

Using the analytical expression of $\check f$, and adding a variable $d_a^{pk}$ to express the term $\frac{1}{b^k u^k_p} \check f_{\mathcal{X}_{\frac{c_a-u_a}{c_a}, \frac{{l_a}^k_p}{c_a}, \frac{c_a - b^k u^k_p}{c_a},  \frac{l_a}{c_a}}} (\frac{b^k}{c_a} {x^k_p},\allowbreak \frac{y_a}{c_a})$, we get the following relaxation of DCMCF: 
\begin{small_align}  \mbox{DCMCF}_\text{env}: & \quad \max \gamma &  \nonumber \\
& \quad \eqref{eq:dcmf:gamma}, \eqref{eq:dcmf:xa},&\\
    &\quad  \sum_{a \in p} d_a^{pk} - \frac{x^k_p}{u^k_p} d^k \leq 0   &\quad \forall k \in K, p \in P^k, \label{eq:dcmcf_env:delay}\\
    &\quad \frac{\frac{x^k_p}{u^k_p}}{c_a - l_a} - 
    d_a^{pk} \leq 0 &  \quad \forall k \in K, p \in P^k, a \in p, \label{eq:dcmcf_env:delay:r}\\
    &\quad  \frac{\frac{x^k_p}{u^k_p}}{c_a - {l_a}^k_p - \frac{ (u_a-{l_a}^k_p) }{u_a - (y_a - b^k x^k_p)}b^k x^k_p} - d_a^{pk} \leq 0 &  \quad \forall k \in K, p \in P^k, a \in p, \label{eq:dcmcf_env:delay:s}\\
    &\quad  \frac{(\frac{x^k_p}{u^k_p})^2}{c_a \frac{x^k_p}{u^k_p} -  y_a + (1-\frac{x^k_p}{u^k_p}) u_a} - d_a^{pk} \leq 0 &  \quad \forall k \in K, p \in P^k, a \in p, \label{eq:dcmcf_env:delay:t} \\
    &\qquad x_p^k \geq 0 &\quad \forall k \in K, p \in P^k,  \label{1eq:dcmcf_env:bound_x}\\
&\qquad  x^k_p \leq u^k_p &\quad \forall k \in K, p \in P^k, \label{2eq:dcmcf_env:bound_x} \\
 & \qquad  y_a \leq u_a  &\quad \forall k \in K, p \in P^k, a \in p, \label{1eq:dcmcf_env:bound_y} \\
& \qquad y_a \geq b^k x_p^k + {l_a}_p^k &\quad \forall k \in K, p \in P^k, a \in p, \label{2eq:dcmcf_env:bound_y}
 \end{small_align}
where Constraints \eqref{eq:dcmcf_env:delay} represent the relaxed delay constraint and Constraints \eqref{eq:dcmcf_env:delay:r}, \eqref{eq:dcmcf_env:delay:s} and \eqref{eq:dcmcf_env:delay:t} represent the respective contributions of functions $r_{\frac{l_a}{c_a}}$, $s_{\frac{c_a-u_a}{c_a}, \frac{{l_a}_p^k}{c_a}}$ and $t_{\frac{c_a-u_a}{c_a}, \frac{c_a-b^k u_p^k}{c_a} }$ to the envelope function. Constraints \eqref{1eq:dcmcf_env:bound_x} to \eqref{2eq:dcmcf_env:bound_y} describe, for $k \in K$, $p \in P^k$ and $a \in A$ the feasible domain for variables $x_p^k$ and $y_a$ on which the convex envelopes of the delay functions are built. 

From Proposition \ref{pro:socp}, one can observe that the $\mbox{DCMCF}_\text{env}$ problem can be formulated as a second-order cone program.

\subsection{Computation of a tight domain}\label{sub:param}

The $\mbox{DCMCF}_\text{env}$ relaxation  relies on the bounds $l_a$, $u_a$ (lower and upper bounds of $y_a$ for $a \in A$), $u^k_p$ (upper bound of $x_p^k$ for $k \in K$, $p \in P^k$) and ${l_a}_p^k$ (lower-bound of $y_a - b^k x^k_p$ for $k \in K$, $p \in P^k$ and $a \in p$).

As suggested in \parencite{ben2006mathematical}, explicit upper bounds for $y_a$ can be computed by exploiting the observation that the load on each arc of the graph is induced by paths whose delays must satisfy some constraints.
 Therefore, for any $a \in A$, there exists a $k \in K$ and $p \in P$ such that $a \in p$ and $\frac{1}{c_a-y_a} + \sum_{a' \in p, a' \neq a}\frac{1}{c_a} \leq d^k$ which allows to compute the following upper bound on $y_a$: $u_a =  \max_{p \ni a} c_a - \frac{1}{d^k - \sum_{a' \in p, a' \neq a}\frac{1}{c_a}}$. On the other hand, $l_a=0$ is a valid lower-bound of $y_a$. 
 
 For every $k \in K$ and $p \in P^k$, the load generated by path $p$, $b^k x^k_p$, is simply bounded by the minimum upper bound of the loads over all arcs used by $p$, which gives the following upper bound on $x_p^k$: $u^k_p = \min_{a \in p} \frac{u_a}{b^k}$.
Finally, ${l_a}_p^k=0$ is a valid lower-bound on the residual load on an arc $a \in A$ for a given commodity $k \in K$ and path $p \in P^k$.

\bigskip

It is also possible to compute tighter  bounds by considering a feasible solution of the DCMCF problem and solving, for every arc $a$, commodity $k$ and path $p$, a convex problem minimizing (resp. maximizing) the higher (resp. lower) bounds for the $y_a$, $x^k_p$, and $y_a - b^k x_p^k$ values under the constraints of the $\mbox{DCMCF}_\text{env}$ problem. 

For instance, given a feasible flow $\gamma_0$ and initial bounds for every variable, a tighter value of $u_{a'}$ for $a' \in A$ can be computed by solving the following convex problem:
\begin{small_align}
    &\max y_{a'} & \nonumber\\
    & \eqref{eq:dcmf:gamma}, \eqref{eq:dcmf:xa}, 
\eqref{eq:dcmcf_env:delay}, \eqref{eq:dcmcf_env:delay:r}, \eqref{eq:dcmcf_env:delay:s}, \eqref{eq:dcmcf_env:delay:t}, \eqref{1eq:dcmcf_env:bound_x},
\eqref{2eq:dcmcf_env:bound_x},
\eqref{1eq:dcmcf_env:bound_y}, \eqref{2eq:dcmcf_env:bound_y},
\nonumber& \nonumber \\
    &\gamma \geq \gamma_0. \nonumber & 
\end{small_align}
Similar problems can be considered to compute tighter values of $u^k_p$, $l_a$, ${l_a}^k_p$ for every $k \in K, p \in P^k$ and $a \in A$.

\section{Alternative formulations of the DCMCF problem}
\label{sec:alternative_relax}

\subsection{A Big-M formulation}
\label{sec:bigM}

The delay constraints~\eqref{eq:dcmf:delay} can be expressed as mixed-integer constraints by introducing binary variables $z_p^k$ indicating whether a path $p$ is used. Using upper bounds ${d_p^k}^{\max}$ on the delay of path $p$ when $z_p^k = x_p^k = 0$, which play the role of Big-$M$ constants, the DCMCF problem can be written as: 
\begin{small_align}
   \mbox{DCMCF}_\text{Big-M}:  
   & \qquad \max \gamma & \nonumber \\
    & \qquad  \eqref{eq:dcmf:gamma}, \eqref{eq:dcmf:xa},  \eqref{1eq:dcmcf_env:bound_x},
\eqref{1eq:dcmcf_env:bound_y}\nonumber  \\
    & \qquad  x^k_p \leq {u^k_p} z^k_p & \forall k \in K, p \in P^k  \label{eq:bigM:z_def}\\
    & \qquad  \sum_{a \in p} \frac{1}{c_a - y_a} \leq d^k z^k_p + {d^k_p}^{max} (1-z^k_p)   & \forall p \in P^k \label{eq:bigM:delay}\\ 
    & \qquad   z^k_p \in \{0,1\} & \forall k \in K, p \in P^k.
\end{small_align}
Note that the strength of the Big-M relaxation strongly depends on the tightness of the ${d_p^k}^{max}$ bounds. 
However, the $\mbox{DCMCF}_\text{Big-M}$ formulation enables a straightforward implementation of branch-and-bound approaches using commercial solvers. Tighter relaxations of the delay constraints such as the one proposed in Section \ref{subsec:env_relaxation} can also be included in the $\mbox{DCMCF}_\text{Big-M}$ problem to provide a more efficient relaxation of DCMF. 

\subsection{A disjunctive-programming-based relaxation}
 \label{sec:hijazi}
We describe a relaxation proposed in \parencite{hijazi2012} and we prove that $\mbox{DCMCF}_\text{env}$ provides tighter bounds than the disjunctive-programming-based relaxation.  

The delay constraint related to a path $p$ is relaxed as follows. The authors consider the convex envelope of ${\Gamma_0}^k_p \cup {\Gamma_1}^k_p$ where  
\begin{small_align*}
    {\Gamma_0}^k_p &= \left\{\left(z^k_p, (y_a)_{a \in p}\right): z^k_p = 0, l_a \leq y_a \leq u_a, \forall a \in p \right\}, 
    \mbox{ and} \\
 {\Gamma_1}^k_p &= \left\{\left(z^k_p, (y_a)_{a \in p}\right): z^k_p = 1, l_a \leq y_a \leq u_a, \forall a \in p, \sum_{a \in p} \frac{1}{c_a-y_a} \leq d^k \right\}.
\end{small_align*}
Then they derive the following relaxation of the DCMCF problem:
\begin{small_align}
  \mbox{DCMCF}_\text{red}
 & \quad \max \gamma  \nonumber \\
& \quad  \eqref{eq:dcmf:gamma}, \eqref{eq:dcmf:xa},  \eqref{1eq:dcmcf_env:bound_x},
\eqref{1eq:dcmcf_env:bound_y}, \eqref{eq:bigM:z_def}, \nonumber \\
    & \quad \sum_{a \in p} \frac{{z^k_p}^2}{c_a z^k_p - {y_a}^k_p} \leq z^k_p d^k \qquad & \forall k \in K, p \in P^k \label{eq:hijazi:delay}\\
    &\quad y_a - (1-z^k_p) u_a \leq {y_a}^k_p \leq y_a \qquad & \forall k \in K, p \in P^k, a \in p\label{eq:hijazi:y_a_k_p}\\
    &\quad z^k_p l_a \leq {y_a}^k_p \leq z^k_p u_a \qquad & \forall k \in K, p \in P^k, a \in p \label{eq:hijazi:y_a_k_p2}\\
    &\quad z^k_p \geq 0 &\forall k \in K, p \in P^k.\label{eq:hijazi:z_pos}
\end{small_align}
where new variables ${y_a}^k_p$ are introduced for every commodity $k \in K$, path $p \in P^k$ and arc $a \in p$.  Notice that this relaxation is the strongest one  among those proposed in \parencite{hijazi2012}.

\begin{proposition}
    $\mbox{DCMCF}_\text{env}$ dominates $DCMCF_{red}$. \label{prop:env_better_than_red}
\end{proposition}

\begin{proof}
Starting from a feasible solution of $\mbox{DCMCF}_\text{env}$,  we build a feasible solution of $\mbox{DCMCF}_\text{red}$ with the same objective value. \\
Let then 
$
\left( (\tilde{x}_p^k)_{k \in K,\, p \in P^k},\; (\tilde{y}_a)_{a \in A},\; \tilde{\gamma} \right)$
be a feasible solution of $\mathrm{DCMCF}_{\text{env}}$, and let 
$
\Bigl( (\bar{x}_p^k)_{k \in K,\, p \in P^k},\; (\bar{z}_p^k)_{k \in K,\, p \in P^k},\; \allowbreak (\bar{y}_a)_{a \in A},\;
{({{}\bar{y}_a}^{k}_p)}_{k \in K, p \in P^k, a \in p},\;
\bar{\gamma} \Bigr)
$
denote the solution that we aim to construct.
We fix the variable values as follows: 
\begin{small_equation*}
\bar{x}_p^k = \tilde{x}_p^k, \quad \bar{y}_a = \tilde{y}_a, 
\quad \bar{z}_p^k = \frac{\bar{x}_p^k}{u^k_p}, \quad {{}\bar{y}_a}^k_p =  \max(\bar{y}_a - (1-\bar{z}^k_p) u_a, ~\bar{z}^k_p l_a), \quad\bar{\gamma} = \tilde{\gamma} .\end{small_equation*}
All constraints of $\mathrm{DCMCF}_{\text{red}}$, except~\eqref{eq:hijazi:delay}, are straightforwardly satisfied.
 As Constraints \eqref{eq:dcmcf_env:delay}, \eqref{eq:dcmcf_env:delay:r} and \eqref{eq:dcmcf_env:delay:t} are satisfied, we necessarily have:
\begin{small_equation*}
\sum_{a \in p} \max( \frac{\bar{z}^k_p}{c_a - l_a}, \frac{{{}\bar{z}^k_p}^2}{c_a \bar{z}^k_p -  \bar{y}_a + (1-\bar{z}^k_p) u_a}) - \bar{z}^k_p d^k \leq 0  \qquad \forall k \in K, p \in P^k.     
\end{small_equation*}
If $\bar{y}_a - (1-\bar{z}^k_p) u_a \geq  \bar{z}^k_p l_a$, then ${{}\bar{y}_a}_p^k=\bar{y}_a - (1-{\bar{z}}^k_p) u_a$, and $\max( \frac{\bar{z}^k_p}{c_a - l_a}, \frac{{{}\bar{z}^k_p}^2}{c_a \bar{z}^k_p -  \bar{y}_a + (1-\bar{z}^k_p) u_a}) = \frac{\bar{z}^k_p}{c_a \bar{z}^k_p - {{}\bar{y}_a}^k_p}$.
Conversely, if $\bar{y}_a - (1-\bar{z}^k_p) u_a \leq  \bar{z}^k_p l_a$, then ${{}\bar{y}_a}^k_p=\bar{z}^k_p l_a$, and 
\begin{small_equation*}
\max( \frac{\bar{z}^k_p}{c_a - l_a}, \frac{{{}\bar{z}^k_p}^2}{c_a \bar{z}^k_p -  y_a + (1-\bar{z}^k_p) u_a}) = \frac{\bar{z}^k_p}{c_a - l_a} = \frac{{{}\bar{z}^k_p}^2}{c_a \bar{z}^k_p - {{}\bar{y}_a}^k_p}.
\end{small_equation*}
This ensures that Constraints \eqref{eq:hijazi:delay} are satisfied in both cases, proving the feasibility of the constructed  solution of $\mbox{DCMCF}_\text{red}$ with $\bar{\gamma} = \tilde{\gamma}$.  
\end{proof}

\subsection{Additional valid constraints \label{sec:add_constraints}}

Two additional sets of valid constraints can be considered in DCMCF relaxations.   
If variables $z^k_p$ are considered, then we can  obviously write:
\begin{small_equation}
    \sum_{p \in P^k} z_p^k \geq 1 \qquad \qquad \forall k \in K. \label{eq:one_path_per_commo}
\end{small_equation}
While the delay Constraints \eqref{eq:dcmf:delay_time_xkp} are not convex, weighting them by the commodity size and aggregating them over all commodities and paths leads to a new convex constraint:
\begin{small_align}
    \sum_{k \in K} b^k \sum_{p \in P^k} \sum_{a \in p} \frac{x^k_p}{c_a-y_a} \leq  \sum_{k \in K} b^k  \sum_{p \in P^k} d^k  x^k_p \quad \Leftrightarrow \quad  
    \sum_{a \in A}\sum_{k \in K, p \in P^k | a \in p}  \frac{b^k x^k_p}{c_a-y_a} \leq  \sum_{k \in K} b^k~\gamma d^k \quad \Leftrightarrow \quad 
    \sum_{a \in A} \frac{y_a}{c_a-y_a} &\leq  \gamma \sum_{k \in K} d^k b^k. \label{eq:aggregated_delays}
\end{small_align}
Observe that  \eqref{eq:aggregated_delays} can be expressed through SOCP.  

These two constraints were tested numerically, but did not provide any significant impact on the problem's resolution. Therefore, the numerical experiments measuring their impact are not reported in Section \ref{sec:num_experiments}.

\section{Lower bound heuristics and performance guarantees}
\label{sec:low}

Having derived upper bounds through the proposed relaxations, we now turn to the question of lower bounds. We begin with a result that may be viewed as negative, and then introduce several heuristics. In particular, we show that an heuristic based on the relaxation $\mathrm{DCMCF}_{\text{env}}$ constitutes an approximation algorithm with provable performance guarantees.

\subsection{A negative results on heuristics} \label{sec:heuristic:negative:result}

This section identifies a general limitation of heuristics applied to problems involving conditional constraints.

We define the term \emph{convex restriction} as the counterpart to a \emph{convex relaxation}: the convex restriction of a domain is a convex subset of that domain. Similarly, the convex restriction of a constraint is a convex constraint that, if satisfied, guarantees the satisfaction of the original constraint.

\medskip

We demonstrate that building a convex restriction of a conditional constraint reduces to a binary discrete choice: either fully activating the constraint or forcing the trigger variable to zero.

\begin{theorem}\label{theorem:heur_neg_result}
Let $E \subseteq \mathbb{R}^+ \times \mathbb{R}^n$ be a set and $f$ be a continuous function $E \rightarrow \mathbb{R}$. Consider the conditional constraint
$
f(x,\mathbf{y}) \leq 0 \quad \text{if} \quad x>0
$
for variables $(x,\mathbf{y}) \in E$. Any convex restriction of this constraint on $E$ is either a restriction of the constraint $f(x,\mathbf{y}) \leq 0$ or a restriction of the constraint $x=0$.
\end{theorem}

\begin{proof}
The proof results directly from Lemma \ref{lemma:heur_neg_result} presented below.
\end{proof}

Theorem \ref{theorem:heur_neg_result} shows that there is no intermediate convex restriction for a conditional constraint and that any heuristic must make a distinct choice for each conditional constraint: it must either enforce the inequality $f(x,\mathbf{y}) \leq 0$ unconditionally or force the value of the trigger variable $x$ to zero.

Consequently, any heuristic approach using convex restrictions is simply a constraint-selection process. To build a feasible solution, the algorithm must divide the conditional constraints into an active set, for which the constraint $f(x,\mathbf{y}) \leq 0$ must be satisfied, and an inactive set, where the variables driving the conditions are forced to zero.

This result proves that searching for continuous convex restrictions for conditional constraints is futile. Instead, heuristic design must focus entirely on the combinatorial selection of the active constraints. 
Once this selection is made, the algorithm solves the resulting problem to obtain a feasible solution. If the original problem only involves convex (conditional) constraints and a convex objective function, this resulting problem is convex and computationally easy to solve.

\begin{lemma}\label{lemma:heur_neg_result}
Let $E \subseteq \mathbb{R}^+ \times \mathbb{R}^n$ be a set and $f: E \rightarrow \mathbb{R}$ be a continuous function.
Any convex subset of $\{(x, \mathbf{y}) \in E \mid x f(x, \mathbf{y}) \leq 0\}$ is either included in $\{(x, \mathbf{y}) \in E \mid f(x, \mathbf{y}) \leq 0\}$ or in $\{(x, \mathbf{y}) \in E \mid x = 0\}$.
\end{lemma}

\begin{proof}
Let us define the following sets:
\begin{align*}
D &= \{(x, \mathbf{y}) \in E \mid x f(x, \mathbf{y}) \leq 0\} \\
D_{\text{on}} &= \{(x, \mathbf{y}) \in E \mid f(x, \mathbf{y}) \leq 0\} \\
D_{\text{off}} &= \{(x, \mathbf{y}) \in E \mid x = 0\}.
\end{align*}
Because $x \geq 0$ for all $(x, \mathbf{y}) \in E$, the condition $x f(x, \mathbf{y}) \leq 0$ is logically equivalent to $x = 0$ or $f(x, \mathbf{y}) \leq 0$. Therefore, $D = D_{\text{on}} \cup D_{\text{off}}$.

Assume there exists a convex set $D' \subseteq D$ that is neither included in $D_{\text{on}}$ nor in $D_{\text{off}}$. There exist two points $p_{\overline{\text{off}}} = (x_{\overline{\text{off}}}, \mathbf{y}_{\overline{\text{off}}}) \in D'  \setminus D_{\text{off}}$ (so $x_{\overline{\text{off}}} > 0$) and $p_{\overline{\text{on}}} = (x_{\overline{\text{on}}}, \mathbf{y}_{\overline{\text{on}}}) \in D' \setminus D_{\text{on}}$ (so $  f(x_{\overline{\text{on}}}, \mathbf{y}_{\overline{\text{on}}}) > 0$)

\medskip

Figure~\ref{man:fig:cond_constr:heur_neg_result} illustrates this proof for $E=[0,1]^2$ and $f(x,y)=x^2+y^2-\frac{1}{4}$. The domain $D$ is the union of three disjoint subsets: $D_{\text{on}} \cap D_{\text{off}}$ in red, $D \setminus D_{\text{off}}$ in orange, and $D \setminus D_{\text{on}}$ in green.  As we will see, a set containing both $p_{\overline{\text{on}}}$ and $p_{\overline{\text{off}}}$ is either non-convex or contains points outside $D$.

\begin{figure}[H]
\centering
\includegraphics[width=0.33\linewidth, page=9, clip, trim={7cm, 5.5cm, 15.cm, 2.5cm}]{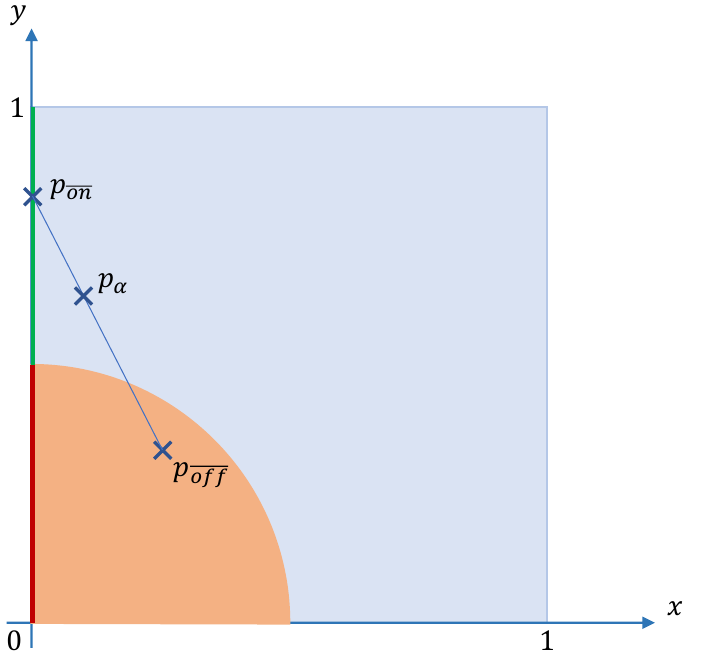}
\caption{Illustration of the proof of Lemma \ref{lemma:heur_neg_result} for $E=[0,1]^2$, and $f(x,y)=x^2+y^2-\frac{1}{4}$.}
\label{man:fig:cond_constr:heur_neg_result}
\end{figure}

\bigskip

Since $p_{\overline{\text{off}}} \in D'  \subseteq D$ and $x_{\overline{\text{off}}} > 0$, we must have $f(x_{\overline{\text{off}}}, \mathbf{y}_{\overline{\text{off}}}) \leq 0$. 
Similarly, since  $p_{\overline{\text{on}}} \in D$ and $f(x_{\overline{\text{on}}}, \mathbf{y}_{\overline{\text{on}}}) > 0$, we must have $x_{\overline{\text{on}}} = 0$.

Because $f$ is continuous on $E$ and $f(x_{\overline{\text{on}}}, \mathbf{y}_{\overline{\text{on}}}) > 0$, there exists a ball $\mathcal{B}$ centered at $p_{\overline{\text{on}}}$ with radius $r > 0$ such that $f(x, \mathbf{y}) > 0$ for all $(x,\mathbf{y}) \in E \cap \mathcal{B}$.

Consider the convex combination $p_\alpha = (x_\alpha, \mathbf{y}_\alpha) = (1 - \alpha) p_{\overline{\text{on}}} + \alpha p_{\overline{\text{off}}}$ with $\alpha = \frac{r}{2 \| p_{\overline{\text{off}}} - p_{\overline{\text{on}}} \|} \in ]0,1]$. 

\begin{itemize}
    \item $\|p_\alpha - p_{\overline{\text{on}}}\|=\alpha\|p_{\overline{\text{off}}}-p_{\overline{\text{on}}}\| < r$, so $p_\alpha \in \mathcal{B}$ and $f(p_\alpha) > 0$.
    \item $x_\alpha = (1 - \alpha) x_{\overline{\text{on}}} + \alpha x_{\overline{\text{off}}} = \alpha x_{\overline{\text{off}}} > 0$.
\end{itemize}

Therefore, $x_\alpha f(x_\alpha, \mathbf{y}_\alpha) > 0$ and $p_\alpha \notin D$ which implies that $p_\alpha \notin D'$ and $D'$ cannot be a convex set. This proves that there is no convex subset of $D$ that is neither a convex subset of $D_\text{on}$ nor $D_\text{off}$.
\end{proof}

\subsection{Active subset heuristics}

Following Theorem \ref{theorem:heur_neg_result}, the main principle of the heuristics we propose is to select a subset of paths $Q^k \subset P^k$ for each $k \in K$, and to replace the constraints \eqref{eq:dcmf:delay} by the constraints $\sum_{a \in p} \frac{1}{c_a - y_a} \leq d^k$ for each $p \in Q^k,$ and to set $x^k_p$ to $0$ for all paths in $P^k \setminus Q^k$. All that remains is to solve a convex problem:
\begin{small_align*}
    \mbox{DCMCF}_\text{subset} &\max \gamma&\\
    &\eqref{eq:dcmf:gamma}, \eqref{eq:dcmf:xa}, \eqref{eq:dcmf:ca}, \eqref{eq:dcmf:pos}&\\
    & \sum_{a \in p} \frac{1}{c_a - y_a} \leq d^k &\qquad \forall k \in K, p \in Q^k.
\end{small_align*}
$\mbox{DCMCF}_\text{subset}$
 is a simple convex optimization problem. The difficulty, of course, is guessing the right subset of paths to select.

\subsubsection{All paths heuristic}\label{sec:heur:all_paths}
Remember that we  assumed that  $\sum_{a \in p} \frac{1}{c_a} < d^k$ for each path $p \in P^k$. One can then simply solve $\mbox{DCMCF}_\text{subset}$ with $Q^k = P^k$ for every commodity $k \in K$ to obtain a feasible solution of the problem with a strictly positive objective value, that will be called $\mathcal{H}_{all~paths}$.

\subsubsection{Greedy heuristic}
The $\mathcal{H}_{all~paths}$ heuristic can be improved in a greedy fashion: the $\mbox{DCMCF}_\text{subset}$ problem is first solved with $Q^k=P^k$ for every $k \in K$, then the path with the smallest $x^k_p$ value is removed. This process is repeated as long as the objective function of the problem increases. This heuristic, called $\mathcal{H}_{greedy}$ is presented in Algorithm \ref{algo:heuristic:greedy}, where $\text{OPT}(\mbox{DCMCF}_{\{Q^k\}_{k \in K}})$ is the optimal value of the $\mbox{DCMCF}_\text{subset}$ problem on the subsets $Q^k$ of $P^k$ for $k \in K$.

\begin{algorithm}
\small
\begin{algorithmic}
    \State $\forall~k \in K$: $Q^k \gets P^k$
    \State $H=0$, $\qquad H^-=~$OPT($\mbox{DCMCF}_{\{Q^k\}_{k \in K}}$)
    \While {$H^-\geq H$}
        \State $(k^-, p^-)=\text{argmin}\{x_p^k~|~x^k_p > 0, \sum_{a \in p} \frac{1}{c_a-y_a}=d^k\}$
        \State $Q^{k^-}=Q^{k^-}\setminus \{p^-\}$
        \State $H = H^-$, $\qquad H^-=~$OPT($\mbox{DCMCF}_{\{Q^k\}_{k \in K}}$)
    \EndWhile
    \State \Return $H$
\end{algorithmic}
    \caption{Greedy Heuristic Algorithm}
    \label{algo:heuristic:greedy}
\end{algorithm}

\subsubsection{Heuristic based on the $\mbox{DCMCF}_\text{env}$ Relaxation}\label{sec:low_hpos}
It also makes sense to use the convex relaxation solution of Section \ref{sec:env} to select the subsets $Q^k$. For example, for a given scalar $S \in \mathbb R^+$, we can select all paths for which $x^k_p \geq \frac{\gamma}{|P^k|}$ in the convex relaxation's optimal solution: $Q^k_S = \{p \in P^k: x^k_p \geq \frac{\gamma}{S|P^k|} \}$.
Let us denote this heuristic by $\mathcal{H}_{threshold}$. If $S \geq 1$ then $|Q^k_S| \geq 1$ for every commodity $k \in K$, and the heuristic $\mathcal{H}_{threshold}$ provides a strictly positive lower-bound for the DCMCF problem.

Before analyzing the performance guarantees of $\mathcal{H}_{threshold}$   in Section \ref{sec:perf_guarantees}, we will present a lower bound of the optimal objective value of DCMCF that can be easily expressed and  polynomially encoded in the size of the problem instance. While this bound is generally rather loose, it will nonetheless be used in the analysis of $\mathcal{H}_{\text{threshold}}$ in the next section.
\begin{lemma}
There exists a feasible solution $\gamma$ of DCMCF such that: 
\begin{small_equation}
    \gamma \geq  \gamma_{low} = \min\limits_{k \in K, p \in P^k} \frac{\frac{1}{\sum_{a \in p} \frac{1}{c_a}} - \frac{1}{d^k}}  {\sum_{k' \in K} b^{k'}}.
    \label{eq:low5}
\end{small_equation}
\begin{proof}
Let $\gamma^{opt}$ be the optimal value of $\gamma$ when DCMCF is solved exactly. Then, there is at least one commodity $k$, and one path $p \in P^k$ such that $\sum_{a \in p} \frac{1}{c_a - y_a} = d^k$. 
Moreover, for any $a \in p$, starting from  $\frac{1}{c_a - y_a} \leq d^k$, we get that $\frac{y_a}{c_a - y_a} \leq y_a d^k$ which can be written as 
$\frac{c_a}{c_a - y_a} \leq 1 +  y_a d^k$. Dividing by $c_a$ leads to 
$\frac{1}{c_a - y_a} \leq \frac{1}{c_a} + y_a \frac{d^k}{c_a}$.   Therefore, $d^k = \sum_{a \in p} \frac{1}{c_a - y_a} \leq \sum_{a \in p} \frac{1}{c_a} + y_a \frac{d^k}{c_a}$.  
By bounding $y_a$ by $\gamma^{opt}(\sum_{k' \in K} b^{k'})$, we get that $d^k \leq \sum_{a \in p} \frac{1}{c_a} + d^k \gamma^{opt}(\sum_{k' \in K} b^{k'}) \sum_{a \in p} \frac{1}{c_a}$,  leading to inequality \eqref{eq:low5}.
\end{proof}
    
\end{lemma}

\subsection{Performance guarantees of $\mathcal{H}_{threshold}$}
\label{sec:perf_guarantees}

We show that $\mathcal{H}_{{threshold}}$, when applied with an appropriately chosen threshold, constitutes an approximation algorithm with provable performance guarantees. 
Because the complete proof is relatively long, we only present the main steps in this document: all intermediate computations can be found in the Appendix \ref{sec:proof:perf_guaranttee}.

\bigskip

\begin{proposition}
    Given a scalar $x_0 \in ]0, 1-\zeta]$ and $\lambda_{\epsilon, x_0}= \frac{4x_0(1-\epsilon)}{(1-\epsilon+x_0)^2}$, we have $f(\lambda x, \lambda y) \leq \lambda \check{f}_{\mathcal{X}_{\epsilon,\theta, \zeta, \eta}} (x,y)$ for any $(x,y) \in \mathcal{X}_{\epsilon,\theta, \zeta, \eta}$ such that $x \geq x_0$ and $\lambda \in [0, \lambda_{\epsilon, x_0}]$. \label{prop:approx:def_lambda}
\end{proposition}

The full proof of Proposition \ref{prop:approx:def_lambda} can be found in Appendix \ref{proof:approx:def_lambda}.
The scalar $\lambda_{\epsilon, x_0}$ can be seen as a shrinking factor transforming a feasible solution of the $\text{DCMCF}_{\text{env}}$ into a feasible solution of the $\text{DCMCF}$ problem that satisfies the delay Constraints \eqref{eq:constr_f} : $\sum_{a \in p} f(\frac{b^k x^k_p}{c_a}, \frac{y_a}{c_a} ) \leq b^k x^k_p \times d^k  \qquad \forall k \in K, p \in P^k$.

\bigskip

Let us introduce additional notation that will be used throughout the remainder of this section: 
$|P| = \max_{k \in K} |P^k| > 1$, 
$u = \max_{a \in A} u_a$, 
and $b = \min_{k \in K} b^k$.

\begin{lemma}
    Given scalars $\gamma >0$ and $S>1$, for any $k \in K$, $p \in P^k$ and $a \in p$ such that $u_a \geq \frac{b^k \gamma}{ S |P^k|} $, we have: 
    \begin{small_equation*}
    {\lambda^{post}_{\gamma,S}} = \frac{4ub \gamma S|P|}{(uS|P|+b \gamma)^2} \leq  \lambda_{1-\frac{u_a}{c_a}, \frac{b^k \gamma}{c_a S |P^k|}}.    \end{small_equation*}        
\label{lemma:approx:lambda_post}
\end{lemma}

Lemma \ref{lemma:approx:lambda_post} translates this generic shrinking factor into the specific parameters of the DCMCF problem to provide a uniform scaling factor $\lambda^{post}_{\gamma,S}$ for all paths exceeding a given threshold. The proof, found in Appendix \ref{proof:approx:lambda_post}, relies on the monotonicity of the shrinking factor with respect to the demand sizes, path counts, and arc load bounds.

\bigskip

\begin{proposition}
    Given a solution of the $\mbox{DCMCF}_\text{env}$ problem with objective value $\tilde \gamma$, there exists a feasible solution of  DCMCF  with objective value $\frac{b \tilde \gamma} {u (|P|-1)+b \tilde \gamma}  \tilde \gamma$. \label{prop:approx:h_pos_guarantee_posterior}
\end{proposition}

Proposition \ref{prop:approx:h_pos_guarantee_posterior} establishes an ex-post performance guarantee by constructing a feasible DCMCF solution from the relaxation. As detailed in Appendix \ref{proof:approx:h_pos_guarantee_posterior}, the proof filters the relaxed solution by setting path usages below a threshold $S^*$ to zero, and scales the flow on the remaining paths by the uniform factor established in Lemma \ref{lemma:approx:lambda_post} to ensure the satisfaction of the delay constraints.

\bigskip

\begin{lemma}
    Given a $\text{DCMCF}$ instance  and a scalar $\alpha > 1$, it is possible to build in polynomial time a relaxation $\mbox{DCMCF}_\text{env}$ with upper bounds $u_a$ of variable $y_a$ for every $a \in A$ and a solution with objective value $\gamma^*$, such that $u_a \leq \alpha  (\sum_{k \in K} b^k) \gamma^*$. \label{lemma:approx:alpha} 
\end{lemma}

Lemma \ref{lemma:approx:alpha} ensures that sufficiently tight upper bounds $u_a$ on the arc loads can be computed. The proof, provided in Appendix \ref{proof:approx:alpha}, relies on an iterative process that repeatedly solves the relaxation until the bounding criterion is met. Combining these tight load bounds with the ex-post guarantee from Proposition \ref{prop:approx:h_pos_guarantee_posterior} allows us to establish the overall approximation ratio of the algorithm:

\begin{theorem}
    Given a scalar $\alpha > 1$, it is possible to build a $\frac{1}{\alpha\beta|K|(|P|-1)+1}$ approximation of the DCMCF problem in polynomial time, with $\beta=\frac{\max_{k \in K} b^k}{\min_{k\in K}b^k}$.
    \label{theorem:approx:h_pos_guarantee_prior}
\end{theorem}

The proof of Theorem \ref{theorem:approx:h_pos_guarantee_prior} is given in Appendix \ref{proof:approx:h_pos_guarantee_prior}, where we consolidate the previous bounds to extract the final approximation factor.

\section{Numerical experiments}\label{sec:num_experiments}

\FloatBarrier
In order to assess the relevance of the new relaxations presented in this work, we implemented them with SCIP 10.0 \parencite{2025scip} in C++. Tests were conducted with $16$ CPUs and 64 GB of RAM.

The reformulations of the DCMCF problem described above were implemented. All of them included the Big-M constraints 
\eqref{eq:bigM:z_def} and \eqref{eq:bigM:delay} introduced in Section \ref{sec:bigM}, and all of them (except the first one) include additional tightening constraints:
\begin{itemize}
    \item $\text{BIG-M}$ : the Big-M formulation introduced in Section \ref{sec:bigM}.
    \item $\text{CONVEX}$ : the Big-M formulation with Constraints \eqref{eq:dcmcf_env:delay} to \eqref{2eq:dcmcf_env:bound_y} from the $\mbox{DCMCF}_\text{env}$ relaxation.
    \item $\text{S}$ : the Big-M formulation, with Constraints \eqref{eq:dcmcf_env:delay}, \eqref{eq:dcmcf_env:delay:s}, and \eqref{1eq:dcmcf_env:bound_x} to \eqref{2eq:dcmcf_env:bound_y} from the $\mbox{DCMCF}_\text{env}$ relaxation, provides a relaxation based on the $s_{\epsilon, \theta}$ function introduced in Section \ref{sec:conv_env:s}.
    \item $\text{T}$ : the Big-M formulation, with Constraints \eqref{eq:dcmcf_env:delay}, and \eqref{eq:dcmcf_env:delay:t}  to
\eqref{2eq:dcmcf_env:bound_y} from the $\mbox{DCMCF}_\text{env}$ relaxation provides a relaxation based on the $t_{\epsilon, \zeta}$ function introduced in Section \ref{sec:conv_env:t} ; this formulation is comparable to the one presented in \parencite{hijazi2012}.
\end{itemize}

The heuristics and relaxations were tested on two sets of instances. The first set is from the Survivable Network Design Library (SND-Lib, \parencite{orlowski2010sndlib}) on which, for every link $(u,v) \in A$, a reverse link $(v,u)$ with capacity $c_{(u,v)}$ was added if it did not exist. Information on these instances is shown in Table \ref{table:bnb_sndlib}.
The second set includes \tmp{126} randomly generated instances (connected  Erdös–Rényi graphs, viewed as  symmetric directed graphs with symmetric arc costs: $c_{(u,v)}=c_{(v,u)}$). These randomly generated instances range from $25$ to $200$ commodities, $10$ to $50$ vertices and $30$ to $475$ arcs.

For each demand $k \in K$, a set $P^k$ of the $10$ shortest paths, weighted by $\frac{1}{c_a}$, was generated through Yen's algorithm \parencite{yen1971}.
For each instance, the same delay $d=1.1 \times \max_{k \in K, p \in P^k} \sum_{a \in p} \frac{1}{c_a}$ was used as the maximum delay for every commodity $k \in K$. This maximum delay ensures that every path could carry some traffic.

\subsection{Heuristics}

We first compare the lower bounds obtained by these $3$ heuristics $\mathcal{H}_{all~paths}$, $\mathcal{H}_{greedy}$ and $\mathcal{H}_{threshold}$ defined in Section \ref{sec:low}. The heuristic $\mathcal{H}_{threshold}$ uses the threshold parameter $S^*=2-\frac{2}{|P|}+\frac{b \tilde \gamma}{u|P|}$ defined in the proof of Proposition \ref{prop:approx:h_pos_guarantee_posterior}.

Figures \ref{fig:heuristics_gap_snd} and \ref{fig:heuristics_gap_rnd} show the gaps between these heuristic and the $\mbox{DCMCF}_\text{env}$ relaxation, which provides the tightest known relaxation to the delay-constraint problem, as shown in Section \ref{sec:env}. Gaps are defined as: $gap=\frac{\gamma^*_H}{\gamma^*_{env}}-1$, with $\gamma^*_H$ the objective value of the heuristic and $\gamma^*_{env}$ the optimal solution of the $\mbox{DCMCF}_\text{env}$ relaxation.
In Figure \ref{fig:heuristics_gap_rnd}, instances are grouped by the number $|K|$ of commodities and presented as a box plot (the box extends from the first quartile to the third quartile of the data, with a line at the median; the whiskers extend from the box to the farthest data point lying within 1.5x the inter-quartile range from the box).

By definition, the $\mathcal{H}_{greedy}$ (in green) heuristic dominates $\mathcal{H}_{all~paths}$ (in red).
While the $\mathcal{H}_{greedy}$ heuristic tends to provide better solutions on smaller instances, its performance seems to deteriorate on larger values of $|K|$ where the $\mathcal{H}_{threshold}$ heuristic appears to be more efficient (as shown in Figure \ref{fig:heuristics_gap_rnd}).

\begin{figure}
\begin{subfigure}[t]{0.5\linewidth}
    \centering
\caption{\small Heuristics gaps for the SND-Lib instances}    \includegraphics[trim={1.3cm 0.cm 2cm 1.0cm},clip,width=1\linewidth]{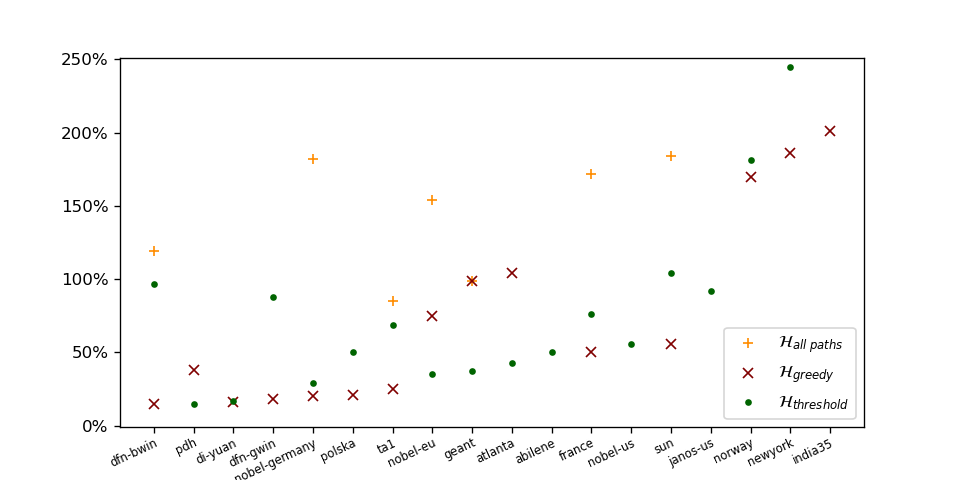}
\label{fig:heuristics_gap_snd}
\end{subfigure}
\hfill
\begin{subfigure}[t]{0.5\linewidth}
    \centering
    \caption{ Heuristics gaps for the random instances grouped by$|K|$}
\includegraphics[trim={1.3cm 0cm 2cm 1.cm},clip, width=1.\linewidth]{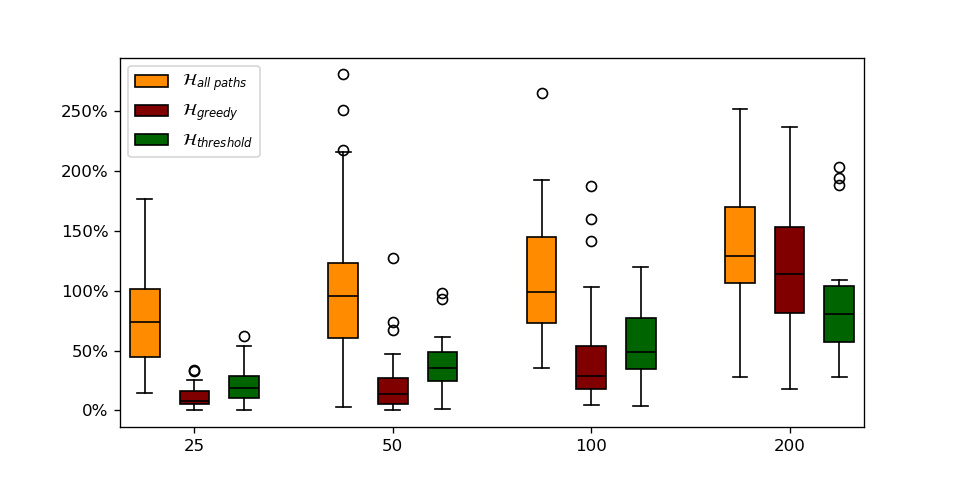}
    \label{fig:heuristics_gap_rnd}
\end{subfigure}
\caption{Gaps between the heuristics presented in Section \ref{sec:low} with respect to the CONVEX relaxation.}
\end{figure}

%\subsection

\subsection{DCMCF relaxations}

The choice of the lower and upper bounds $u_a$, $l_a$, $u_p^k$, and $l_p^k$ significantly affects the strength of the resulting relaxations.
As described in Section \ref{sub:param}, bounds  can be computed through a closed-formula or a convex optimization problem. 
In practice, the  bounds ${l_a}_p^k$ are set to $0$, while  $l_a, u_a$ and $u_p^k$ are computed by solving the corresponding convex-optimization problems.  Notice that even when we are considering the bounds related to the $\text{BIG-M}$ (resp. $\text{S}$ or $\text{T}$)  formulation, we use the strongest relaxation (i.e., $\mbox{DCMCF}_\text{env}$) to compute them. 

We compare the convex relaxations of the DCMCF problem associated with the formulations $\text{BIG-M}$, $\text{S}$, $\text{T}$ and  $\text{CONVEX}$ described above.
As $\text{S}$, $\text{T}$ and $\text{CONVEX}$ include the constraints of the $\text{BIG-M}$ relaxation, they necessarily provide tighter upper bounds for the problem. Similarly, all constraints from the $\text{S}$ and $\text{T}$ relaxations are included in the  $\text{CONVEX}$ relaxation, meaning that the upper bound of the latter is necessarily tighter than the upper bound of the former.

Figure \ref{fig:relaxations} provides comparisons of these upper bounds, showing the relative gains of $\text{S}$, $\text{T}$ and $\text{CONVEX}$ versus the $\text{BIG-M}$ relaxation. The relative-gain is defined as the difference in objective value between $\text{BIG-M}$ and the considered relaxation, divided by the difference between the objective value of the $\text{BIG-M}$ relaxation and the best of the $\mathcal{H}_{threshold}$, $\mathcal{H}_{greedy}$ or $\mathcal{H}_{all~paths}$ heuristics: $rel.~gain_{relax}=\frac{\gamma^*_{BIG-M} - \gamma^*_{relax}}{\gamma^*_{BIG-M}-\gamma^*_{best~\mathcal{H}}}$ where $\gamma^*_{relax}$ is the optimal value of the relaxation $relax$ and $\gamma^*_{best~\mathcal{H}}$ is the best lower-bound obtained by the heuristics described above. 
The name of the SND-Lib instances is provided in Figure \ref{fig:relaxations_snd} - but the name of the \tmp{$126$} random instances is not shown in Figure \ref{fig:relaxations_rnd}. In both figures, instances are ordered by value of the $\text{CONVEX}$ relative gain.

As proved in Section \ref{sec:env}, the $\text{CONVEX}$ reformulation indeed dominates both $\text{S}$ and $\text{T}$. 
The absence of domination between the relaxations $\text{S}$ and $\text{T}$ is consistent with their definition and the fact that functions $s_{\epsilon, \gamma}$ and $t_{\epsilon, \zeta}$ dominate each other in different subsets of the feasible region; however, on the randomly-generated instances tested, the $S$ relaxation provides tighter upper bounds than the $T$ relaxation on \tmp{$123$} of the \tmp{$126$} instances tested.

\begin{figure}    
\begin{subfigure}[t]{0.5\textwidth}
    \centering
    \caption{SND-Lib instances}
    \includegraphics[trim={1.2cm 0.cm 1.5cm 1cm},clip, width=1\linewidth]{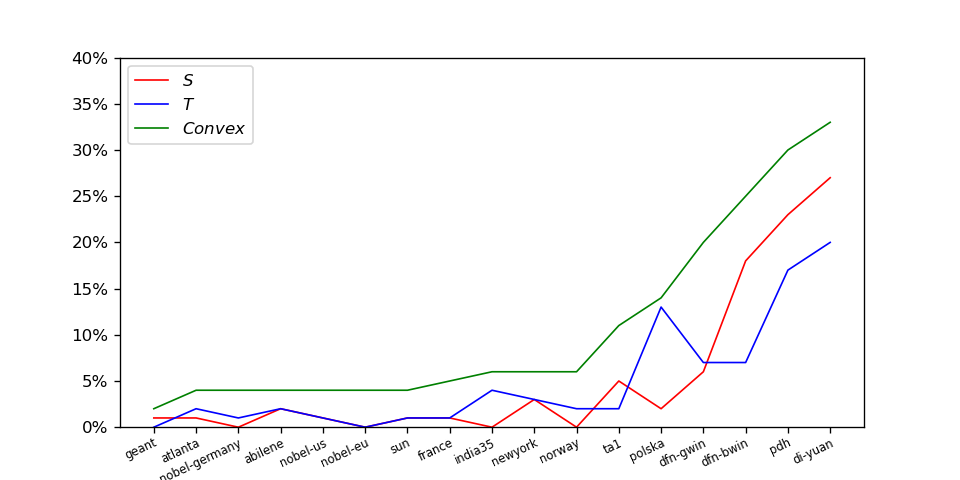}
    \label{fig:relaxations_snd}
\end{subfigure}
\hfill 
\begin{subfigure}[t]{0.5\textwidth}
    \centering
    \caption{Randomly-generated instances}
    \includegraphics[trim={1.2cm 0.cm 1.5cm 1cm},clip,width=1\linewidth]{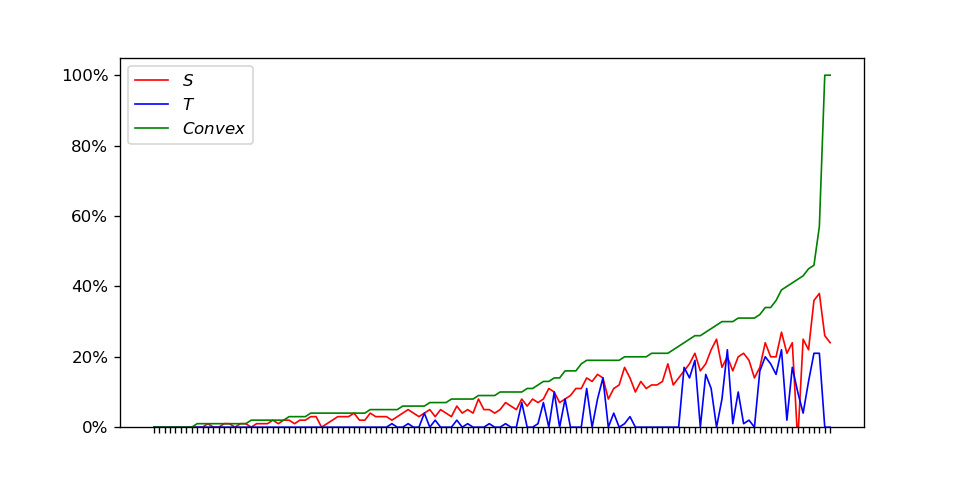}
    
    \label{fig:relaxations_rnd}
\end{subfigure}
\caption{Relative gains of the relaxations $\text{CONVEX}$, $\text{S}$ and $\text{T}$ versus the $\text{BIG-M}$ relaxation.}
\label{fig:relaxations}
\end{figure}

\subsection{DCMCF exact solving}

The DCMCF was solved using a branch-and-bound algorithm. The 4 relaxations presented above ($\text{BIG-M}$, $\text{S}$, $\text{T}$ and $\text{CONVEX}$) were implemented and compared. Tables \ref{table:bnb_sndlib} shows, for each  SND-lib instance, the respective performances of these formulations: either the solving time if an optimality gap of $0.1\%$ was reached within $1$ hour, or the optimality gap if it was above $0.1\%$ after this time limit. For the randomly generated instances, the proportion of random instances solved after a given solving time or, for instances not solved after one hour, the residual gap, are shown in Figure \ref{fig:perf_graphs_exact_rnd}.

The $\text{CONVEX}$ formulation provides a tighter relaxation of the problem than the lighter $\text{BIG-M}$ relaxation, at the expense of its solving time. These effects balance each other on instances of the SND-Lib, where the $\text{CONVEX}$ formulation outperformed the $\text{BIG-M}$ one (either in terms of gap after one hour or solving time if optimality is reached before that time limit)  in $11$ instances over $19$. On the randomly generated instances, however, the lighter $\text{BIG-M}$ relaxation proves more efficient on \tmp{$77$} of the \tmp{$126$} instances tested.

The $\text{S}$ and $\text{T}$ formulations provide good trade-offs between tightness and relaxation solving time, allowing them to outperform both the $\text{BIG-M}$ (on $12$ and $13$ instances of the SND-Lib and \tmp{$72$} and \tmp{$64$} randomly-generated instances, respectively). The dominance of the $\text{S}$ relaxation, illustrated in Figure \ref{fig:relaxations}, does not materialize when solving the exact DCMCF problem in a branch-and-bound scheme: the $\text{S}$ relaxation is more efficient than $\text{T}$ in $8$ of the $19$ SND-Lib instances and $72$ of the $126$ random instances.

\begin{table}[ht]
\scriptsize
\centering
\begin{tabular}{c|c c c|c c c c}
Instance	&	$|K|$	&	$|A|$	&	$|V|$	&	$\text{BIG-M}$		&	$\text{S}$		&	$\text{T}$		&	$\text{CONVEX}$		\\ \hline
pdh	&	$24$	&	$68$	&	$11$	&	$\boldsymbol	{(250)}$	&	$	{(960)}$	&	$	{(582)}$	&	$	{(614)}$	\\
di-yuan	&	$22$	&	$84$	&	$11$	&	$\boldsymbol	{(188)}$	&	$	{(462)}$	&	$	{0.6\%}$	&	$	{(603)}$	\\
polska	&	$66$	&	$36$	&	$12$	&	$	{15.9\%}$	&	$	{16.9\%}$	&	$\boldsymbol	{13.5\%}$	&	$	{16\%}$	\\
nobel-us	&	$91$	&	$42$	&	$14$	&	$	{29.6\%}$	&	$	{74.3\%}$	&	$	{58.5\%}$	&	$\boldsymbol	{21.3\%}$	\\
abilene	&	$132$	&	$30$	&	$12$	&	$	{35.8\%}$	&	$	{13.5\%}$	&	$\boldsymbol	{6.5\%}$	&	$	{29.4\%}$	\\
nobel-germany	&	$121$	&	$52$	&	$17$	&	$	{18.1\%}$	&	$	{16.5\%}$	&	$\boldsymbol	{16.3\%}$	&	$	{17.5\%}$	\\
dfn-bwin	&	$90$	&	$90$	&	$10$	&	$	{7.9\%}$	&	$\boldsymbol	{6.6\%}$	&	$	{9.8\%}$	&	$	{10.1\%}$	\\
atlanta	&	$210$	&	$22$	&	$15$	&	$	{55.4\%}$	&	$\boldsymbol	{34.5\%}$	&	$	{42.1\%}$	&	$	{59.2\%}$	\\
dfn-gwin	&	$110$	&	$47$	&	$11$	&	$	{17.6\%}$	&	$	{16.9\%}$	&	$	{17.5\%}$	&	$\boldsymbol	{16.6\%}$	\\
sun	&	$67$	&	$102$	&	$27$	&	$	{35.1\%}$	&	$	{43\%}$	&	$\boldsymbol	{35\%}$	&	$	{54.5\%}$	\\
newyork	&	$240$	&	$98$	&	$16$	&	$	{104\%}$	&	$	{104\%}$	&	$	{104\%}$	&	$\boldsymbol	{99.6\%}$	\\
france	&	$300$	&	$90$	&	$25$	&	$	{51\%}$	&	$	{52.4\%}$	&	$\boldsymbol	{49.5\%}$	&	$	{50\%}$	\\
nobel-eu	&	$378$	&	$82$	&	$28$	&	$	{113.8\%}$	&	$\boldsymbol	{49.3\%}$	&	$	{58.2\%}$	&	$	{50.9\%}$	\\
ta1	&	$396$	&	$102$	&	$24$	&	$	{27.4\%}$	&	$\boldsymbol	{26.3\%}$	&	$	{26.9\%}$	&	$	{26.8\%}$	\\
geant	&	$462$	&	$72$	&	$22$	&	$	{97.2\%}$	&	$	{13.9\%}$	&	$\boldsymbol	{9.4\%}$	&	$	{98.8\%}$	\\
norway	&	$702$	&	$102$	&	$27$	&	$	{180.8\%}$	&	$\boldsymbol	{174.1\%}$	&	$	{180.8\%}$	&	$	{180.8\%}$	\\
india35	&	$595$	&	$160$	&	$35$	&	$	{179.7\%}$	&	$\boldsymbol	{158.4\%}$	&	$	{174.8\%}$	&	$	{212.8\%}$	\\
germany50	&	$662$	&	$88$	&	$50$	&	$	{128.6\%}$	&	$	{128.6\%}$	&	$	{128\%}$	&	$\boldsymbol	{80\%}$	\\
cost266	&	$1332$	&	$114$	&	$37$	&	$	{68.7\%}$	&	$	{67.3\%}$	&	$\boldsymbol	{64.4\%}$	&	$	{68.6\%}$
\end{tabular}
\caption{Branch-and-bound results for the SND-Lib instances}
\label{table:bnb_sndlib}
\end{table}

\begin{figure}[ht]
\begin{subfigure}[t]{0.5\textwidth}
    \centering
    \caption{Performance Graph: solving time}
    \includegraphics[trim={1.3cm 0.5cm 1.5cm 1cm},clip, width=1\linewidth]{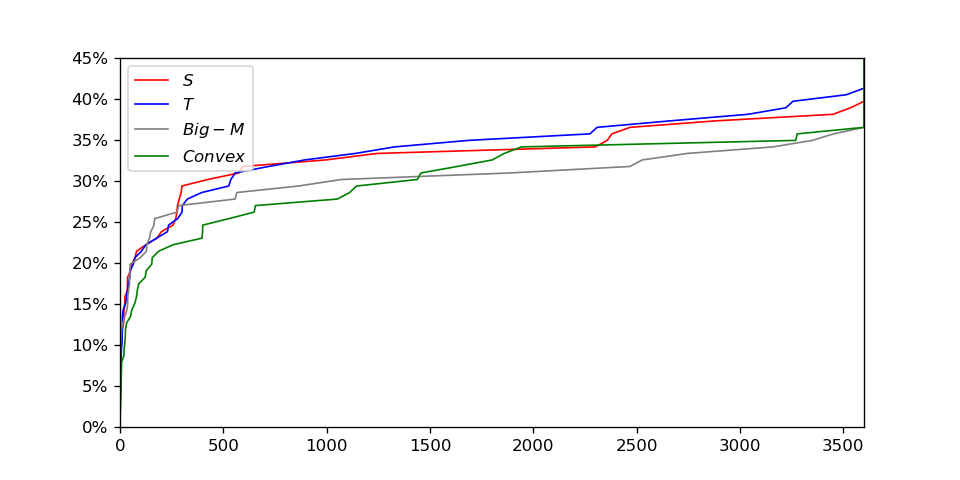}
    \label{fig:perf_dur_rnd}
\end{subfigure}
\hfill 
\begin{subfigure}[t]{0.5\textwidth}
    \centering
    \caption{Performance Graph: gap after 1h}
\includegraphics[trim={1.3cm 0.5cm 1.5cm 1cm},clip, width=1\linewidth]{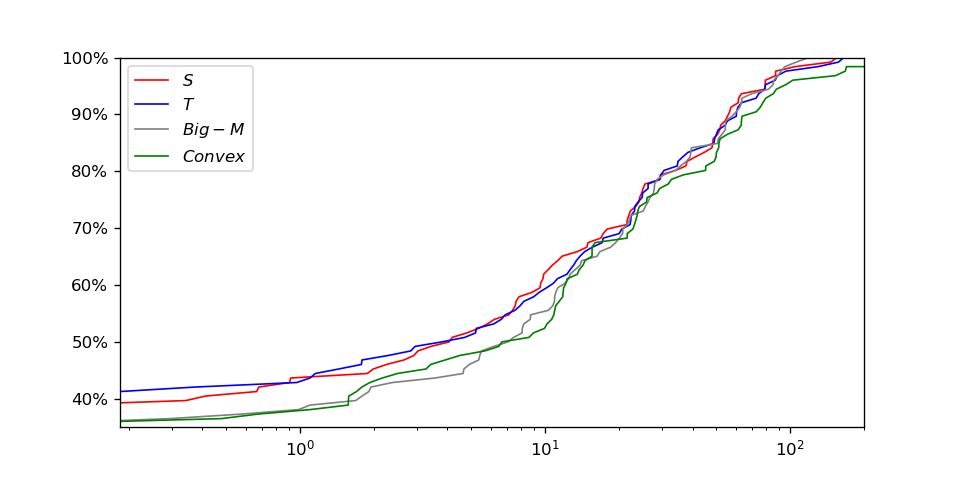}
    \label{fig:perf_gap_rnd}
\end{subfigure}
\caption{Performance graphs for the exact resolution of the DCMCF problem on the randomly generated instances. Left: proportion of instances solved in less than a solving time (for instances solved in less than $1$h); right: proportion of instances with a gap below a specified value after $1$h.}
\label{fig:perf_graphs_exact_rnd}
\end{figure}

\section{Conclusion}

The Delay-Constrained Multi-Commodity Flow (DCMCF) problem was known to be NP-hard. 
In this paper, we proved that the problem is strongly NP-hard, 
even in the case of three commodities that share the same source and destination, 
and where each commodity is allowed to use every possible path.\\
We also introduced a new relaxation of the conditional delay constraints and showed that it dominates the relaxations based on disjunctive programming. \\
Moreover, we described a non-trivial approximation algorithm based on the proposed relaxation and established a provable performance guarantee for it.\\
In the process of developing lower-bound heuristics, we proved that the use of certain convex over-estimators does not provide any advantage in the broader context of conditional constraints. In other words, employing such a convex over-estimator when dealing with a constraint of the form $x g(x,y) \leq 0$, with $x \geq 0$, is equivalent to either fixing $x$ to $0$ or imposing the constraint $g(x,y) \leq 0$.

While this study improves the global understanding of the delay-constrained multi-commodity flow problem, several questions still remain open. For instance, the exact complexity of the problem with a single commodity remains unknown; further improvements to the problem's relaxations could be investigated in order to develop tractable optimization methods for larger instances.
 Future work could also focus on developing approximation algorithms with improved performance guarantees or on refining the analysis of the algorithm presented in this paper.

\begin{appendices}    

\section{Proof of the DCMCF complexity}

\subsection{ Proof of Lemma \ref{lemma:complex:path_k2}}\label{proof:complex:path_k2}

\begin{proof}
The set of all possible paths from $s$ to $t$ is:
$$ \begin{aligned} &  \{s \rightarrow u_i \rightarrow v_j \rightarrow t\}_{i \in I, j \in J} \cup \{s \rightarrow u_i \rightarrow v_j \rightarrow v_j^1 \rightarrow \dots \rightarrow v_j^N \rightarrow t\}_{i \in I, j \in J} \\ &\cup \{s \rightarrow v_j \rightarrow v_j^1 \rightarrow \dots \rightarrow v_j^N \rightarrow t\}_{j \in J} \cup \{s \rightarrow u_i \rightarrow t\}_{i \in I} \cup \{s \rightarrow v_j \rightarrow t\}_{j \in J} \cup \{s \rightarrow v_j^N \rightarrow t\}_{j \in J}. \end{aligned} $$

We must prove that commodity $k_2$ cannot follow paths from the first three subsets of this union.

For $i \in I$ and $j \in J$, the minimum delay of path $s \rightarrow u_i \rightarrow v_j \rightarrow t$ (obtained if the load of all arcs is zero) is:
$$ d_{s \rightarrow u_i \rightarrow v_j \rightarrow t} = \frac{1}{m w_i + (m-1)W + 2} + \frac{1}{w_i + 2W} + \frac{1}{\frac{W}{m} + 2}. $$
Using the assumption that $w_i \leq \frac{W}{m}$ and $W = \frac{1}{4n^2}$, we can bound the delay of the middle arc $(u_i,v_j)$. We have:
$$ \frac{1}{w_i + 2W} \geq \frac{1}{\frac{W}{m} + 2W} = \frac{m}{W(1 + 2m)} = \frac{4 m n^2}{2m + 1}. $$
Because we assume $n \geq 6$ and $m \geq 2$, we have $n^2 \geq 36$ and $\frac{m}{2m+1} \geq \frac{2}{5}$. Therefore
$ \frac{1}{w_i + 2W} \geq 4 \times \frac{2}{5} \times 36 > 2$. 
Since the delay of this single arc strictly exceeds $2$, we have $d_{s \rightarrow u_i \rightarrow v_j \rightarrow t} > 2=b_2$, ruling out this path for commodity $k_2$.

\smallskip

Similarly, the minimum delay through a path $s \rightarrow u_i \rightarrow v_j \rightarrow v_j^1 \rightarrow \dots \rightarrow v_j^N \rightarrow t$ is strictly greater than $2$, since it contains the arc $(u_i,v_j)$, whose zero-load delay we just showed exceeds $2$. Notice that the minimum delay through the chain of arcs $v_j \rightarrow v_j^1 \rightarrow \dots \rightarrow v_j^N$ exceeds $N\frac{1}{W\frac{(n+1)(n-3)}{n}+1}=\frac{1}{W}\frac{1}{\frac{n^2-2n-3}{4n^3}+1} > \frac{1}{W}\frac{1}{2}$ and therefore the delay of a path $s \rightarrow u_i \rightarrow v_j \rightarrow v_j^1 \rightarrow \dots \rightarrow v_j^N \rightarrow t$ strictly exceeds $d_L=2+\frac{1}{2W}$, also preventing $k_2$ from using it.

\smallskip

Finally, for $j \in J$, paths $s \rightarrow v_j \rightarrow v_j^1 \rightarrow \dots \rightarrow v_j^N \rightarrow t$ have a minimum delay:
$$ d_{s \rightarrow v_j \rightarrow v_j^1 \rightarrow \dots \rightarrow v_j^N \rightarrow t} = \frac{1}{2} + (N-1) \frac{1}{W\frac{(n+1)(n-3)}{n} + 1} + \frac{1}{W\frac{(n+1)(n-3)}{n} + 2}. $$
Because $n \geq 6$ and $W = \frac{1}{4n^2}$, the capacity $W\frac{(n+1)(n-3)}{n} + 1$ is strictly less than $2$. As a result, the delay of each of the $N-1$ internal arcs strictly exceeds $\frac{1}{2}$. Since $N = 4n^2 \geq 144$, the total delay is vastly greater than $2$. Therefore, commodity $k_2$ cannot use this path.
\end{proof}

\subsection{Proof of Lemma \ref{lemma:complex:uses} }\label{proof:complex:uses}

\begin{proof}
From Lemma \ref{lemma:complex:path_k2}, the set of feasible paths for $k_2$ contains $n + 2m$ paths. Let us suppose that one of these paths is not used by $k_2$ (i.e., it carries a load of zero). Because the total size of $k_2$ is $b_2 = n + 2m$, the average load on the remaining paths must be strictly greater than $1 + \delta$, where $\delta = \frac{1}{n+2m} = \frac{3}{5n}$. Consequently, at least one path must carry a load strictly greater than $1 + \delta$. We show that the delay of such a path is strictly greater than $d_2 = 2$, which contradicts the feasibility of the flow.

\begin{itemize}
    \item \textbf{Case 1.} Assume that the load through a path $\{s \rightarrow u_i \rightarrow t\}_{i \in I}$ exceeds $1 + \delta$. The delay is:
    $$d_{s \rightarrow u_i \rightarrow t} \geq \frac{1}{c_{(s,u_i)} - (1+\delta)} + \frac{1}{c_{(u_i,t)} - (1+\delta)} = \frac{1}{m w_i + (m-1)W + 1 - \delta} + \frac{1}{1 - \delta}.$$
    Using $w_i \leq \frac{W}{m}$, we have $m w_i + (m-1)W \leq W + mW - W = mW=\frac{n/3}{(2n)^2}=\frac{1}{12n}$. Thus, because $\frac{1}{1-e}\geq 1+e$ for $e<1$:
    $$d_{s \rightarrow u_i \rightarrow t} \geq \frac{1}{1+\frac{1}{12n}-\frac{3}{5n}}+\frac{1}{1-\frac{3}{5n}}\geq (1 - \frac{1}{12n}+\frac{3}{5n}) + (1 + \frac{3}{5n}) =2+2\frac{3}{5n}-\frac{1}{12n}>2$$ ensuring that commodity $k_2$ cannot satisfy its delay if a load $1+\delta$ crosses path $s \rightarrow u_i \rightarrow t$.

    \item \textbf{Case 2.} Assume that the load through a path $\{s \rightarrow v_j \rightarrow t\}_{j \in J}$ exceeds $1 + \delta$. The delay is:
    $$d_{s \rightarrow v_j \rightarrow t} \geq \frac{1}{c_{(s,v_j)} - (1+\delta)} + \frac{1}{c_{(v_j,t)} - (1+\delta)} = \frac{1}{1 - \delta} + \frac{1}{\frac{W}{m} + 1 - \delta}.$$
    Using the same reasoning as above, we get: 
    $$d_{s \rightarrow v_j \rightarrow t} \geq \frac{1}{1-\frac{3}{5n}} + \frac{1}{1+\frac{3}{4n^3}-\frac{3}{5n}} \geq (1+\frac{3}{5n})+(1-\frac{3}{4n^3}+\frac{3}{5n})=1+2\frac{3}{5n}-\frac{3}{4n^3}>2$$ ensuring that commodity $k_2$ cannot satisfy its delay if a load $1+\delta$ crosses path $s \rightarrow v_j \rightarrow t$.

    \item \textbf{Case 3.} Assume that the load through a path $\{s \rightarrow v_j^N \rightarrow t\}_{j \in J}$ exceeds $1 + \delta$. The delay is:
    $$d_{s \rightarrow v_j^N \rightarrow t} \geq \frac{1}{c_{(s,v_j^N)} - (1+\delta)} + \frac{1}{c_{(v_j^N,t)} - (1+\delta)} = \frac{1}{1 - \delta} + \frac{1}{W \frac{(n+1)(n-3)}{n} + 1 - \delta}.$$
    Following the same reasoning as the other cases, we get:
    \begin{align*}
    d_{s \rightarrow v_j^N \rightarrow t} &\geq \frac{1}{1-\frac{3}{5n}}+\frac{1}{1+\frac{(n+1)(n-3)}{4n^3}-\frac{3}{5n}}\\&
    \geq (1+\frac{3}{5n})+(1-\frac{(n+1)(n-3)}{4n^3}+\frac{3}{5n})=2+2\frac{3}{5n}-\frac{(n+1)(n-3)}{4n^3}\\
    &>2.\end{align*}
\end{itemize}

Therefore, commodity $k_2$ cannot route more than $1 + \delta$ on any of these paths. Since the total load is $n+2m$, it must use each path in $\{s \rightarrow u_i \rightarrow t\}_{i \in I} \cup \{s \rightarrow v_j \rightarrow t\}_{j \in J} \cup \{s \rightarrow v_j^N \rightarrow t\}_{j \in J}$.
\end{proof}

\subsection{Proof of Lemma \ref{lemma:complex:ineq_ysu}} \label{proof:complex:ineq_ysu}

\begin{proof}
Let us consider the delay of arc $(u_i,t)$. Using the inequality $\frac{1}{1-e} \geq 1+e$ for $e < 1$ (with a strict inequality if $e \neq 0$ and equality if $e=0$), we get:
$$
d_{(u_i,t)} = \frac{1}{c_{(u_i,t)}-y_{(u_i,t)}}=\frac{1}{2-(1+\epsilon_i)}=\frac{1}{1-\epsilon_i}
\qquad \implies 
\begin{aligned}
&d_{(u_i,t)} > 1+\epsilon_i \quad && \text{if}\quad \epsilon_i \neq 0 \quad (\text{i.e., if } i \in Q_I)\\
&d_{(u_i,t)} = 1 \quad && \text{if}\quad \epsilon_i = 0 \quad (\text{i.e., if } i \in I \setminus Q_I).\\
\end{aligned}
$$

Using the same inequality on the delay of the arc $(s,u_i)$ we get:
\begin{align*}
d_{(s,u_i)} &= \frac{1}{c_{(s,u_i)}-y_{(s,u_i)}}=\frac{1}{mw_i+(m-1)W+2-y_{(s,u_i)}} \geq 1 + y_{(s,u_i)} - (mw_i+(m-1)W+1).
\end{align*}

Because commodity $k_2$ must use the path $s \rightarrow u_i \rightarrow t$, the total delay of this path is less than or equal to $d_2=2$. Therefore:
\begin{itemize}
    \item if $i \in Q_I$:
    $ d_2=2 \geq d_{(s,u_i)} + d_{(u_i,t)} > 1 + \epsilon_i + 1 + y_{(s,u_i)} - (m w_i + (m-1)W + 1)$
    which implies that: $$y_{(s,u_i)} < m w_i + (m-1)W + 1 - \epsilon_i \qquad \forall i \in Q_I.$$
    \item if $i \in I \setminus Q_I$, we have $d_{(u_i,t)}=1$ and therefore $d_{(s,u_i)} \leq 1$, which leads to: $$y_{(s,u_i)} \leq m w_i + (m-1)W + 1 \qquad \forall i \in I \setminus Q_I.$$
\end{itemize}
\end{proof}

\subsection{Proof of Lemma \ref{lemma:complex:ineq_yvt}} \label{proof:complex:ineq_yvt}

\begin{proof}
The reasoning is identical to the proof of Lemma \ref{lemma:complex:ineq_ysu}. We apply the inequality $\frac{1}{1-e} \geq 1+e$ (valid for $e < 1$, with a strict inequality when $e \neq 0$ and equality for $e=0$) to the delays of the arcs forming the 2-hop paths.

Consider the path $s \rightarrow v_j \rightarrow t$. The delay of the first arc is $d_{(s,v_j)} = \frac{1}{c_{(s,v_j)}-y_{(s,v_j)}}=\frac{1}{2 - (1+\epsilon_j)}$. This gives $d_{(s,v_j)} > 1 + \epsilon_j$ if $j \in Q_J$, and $d_{(s,v_j)} = 1$ otherwise. For the second arc, applying the same inequality gives:
$$d_{(v_j,t)} = \frac{1}{c_{(v_j,t)}-y_{(v_j,t)}}=\frac{1}{\frac{W}{m} + 2 - y_{(v_j,t)}} \geq 1 + y_{(v_j,t)} - \left(\frac{W}{m}+1\right).$$
Because commodity $k_2$ must respect the maximum delay constraint $d_2 = 2 \geq  d_{(s,v_j)} + d_{(v_j,t)}$ on this path, we get: 
\begin{itemize}
\item For $j \in Q_J$, the strict inequality on the first arc implies $2 > 1 + \epsilon_j + 1 + y_{(v_j,t)} - \left(\frac{W}{m}+1\right)$, which leads to $y_{(v_j,t)} < \frac{W}{m} + 1 - \epsilon_j$. 

\item For $j \in J \setminus Q_J$, the equality on the first arc implies $2 \geq 1 + 1 + y_{(v_j,t)} - \left(\frac{W}{m}+1\right)$, leading to $y_{(v_j,t)} \leq \frac{W}{m} + 1$.
\end{itemize}
\medskip

The same reasoning applies to the path $s \rightarrow v_j^N \rightarrow t$. The delay of the first arc $d_{(s,v_j^N)} = \frac{1}{c_{(s,v_j^N)}-y_{(s,v_j^N)}}=\frac{1}{2 - (1+\epsilon'_j)}$ is strictly greater than $1 + \epsilon'_j$ for $j \in Q'_J$ and equals $1$ if $j \in J \setminus Q'_J$. The delay of the second arc satisfies:
$$d_{(v_j^N,t)} = \frac{1}{c_{(v_j^N,t)}-y_{(v_j^N,t)}}= \frac{1}{W\frac{(n+1)(n-3)}{n} + 2 - y_{(v_j^N,t)}} \geq 1 + y_{(v_j^N,t)} - \left(W\frac{(n+1)(n-3)}{n}+1\right).$$
Enforcing the maximum delay of $2$ across the entire path guarantees that $y_{(v_j^N,t)} < W\frac{(n+1)(n-3)}{n} + 1 - \epsilon'_j$ for any $j \in Q'_J$, and less than or equal to $W\frac{(n+1)(n-3)}{n} + 1$ otherwise.
\end{proof}

\subsection{Proof of Lemma \ref{lemma:complex:load_k2}} \label{proof:complex:load_k2}

\begin{proof}
We will consider the sum of the loads of all arcs exiting $s$ and entering $t$ to show that successfully routing a throughput of $\gamma = 1$ requires $Q_I \cup Q_J \cup Q_J' = \emptyset$. 

From Lemmas \ref{lemma:complex:ineq_ysu} and \ref{lemma:complex:ineq_yvt}, and the definitions of the $\epsilon$ variables, we can bound the load on each arc of the 2-hop paths:
\begin{itemize}
    \item For a given $i \in I$, the upper bounds on the arc loads are $y_{(s,u_i)} \leq m w_i + (m-1)W + 1 - \epsilon_i$ and $y_{(u_i,t)} \leq 1 + \epsilon_i$. The first bound is strict if $i \in Q_I$. Summing over $i \in I$ gives:
    $$\sum_{i \in I} (y_{(s,u_i)} + y_{(u_i,t)}) \leq mW + n(m-1)W + 2n = b_L + b_H + 2n$$
    with this inequality being strict if $Q_I \neq \emptyset$.

    \item For a given $j \in J$, the upper bounds on the arc loads are $y_{(s,v_j)} \leq 1 + \epsilon_j$ and $y_{(v_j,t)} \leq \frac{W}{m} + 1 - \epsilon_j$. The second bound is strict if $j \in Q_J$. Summing over $j \in J$ gives:
    $$\sum_{j \in J} (y_{(s,v_j)} + y_{(v_j,t)}) \leq W + 2m = b_L + 2m$$
    with this inequality being strict if $Q_J \neq \emptyset$.

    \item For a given $j \in J$, the upper bounds on the arc loads are $y_{(s,v_j^N)} \leq 1 + \epsilon'_j$ and $y_{(v_j^N,t)} \leq W\frac{(n+1)(n-3)}{n} + 1 - \epsilon'_j$. The second bound is strict if $j \in Q'_J$. Summing over $j \in J$ gives:
    $$\sum_{j \in J} (y_{(s,v_j^N)} + y_{(v_j^N,t)}) \leq m W\frac{(n+1)(n-3)}{n} + 2m = b_H + 2m$$
    with this inequality being strict if $Q'_J \neq \emptyset$.
\end{itemize}

Because there are no arcs exiting $s$ or entering $t$ other than those analyzed above, we can sum them to obtain the global bound on the network's entry and exit flows:
$$ \sum_{a \in \delta^+(s)} y_a + \sum_{a \in \delta^-(t)} y_a \leq 2(b_L + b_H + b_2) $$
with a strict inequality if $Q_I \cup Q_J \cup Q_J' \neq \emptyset$.

To successfully route the complete demands for all three commodities, the total flow leaving $s$ and entering $t$ must be exactly $2(b_L + b_H + b_2)$. Therefore, we must have $Q_I \cup Q_J \cup Q_J' = \emptyset$, meaning all $\epsilon_i$, $\epsilon_j$, and $\epsilon'_j$ terms are exactly zero.

It also follows that all the upper bounds on the arc loads hold as strict equalities. We can finally deduce that $y_{(s,u_i)} = m w_i + (m-1)W + 1$, $y_{(v_j,t)} = \frac{W}{m}+1$, $y_{(v_j^N,t)} = W\frac{(n+1)(n-3)}{n}+1$, $y_{(u_i,t)} = 1$, $y_{(s,v_j)} = 1$, and $y_{(s,v_j^N)} = 1$. We also conclude that the delay on each of these arcs is exactly equal to 1.
\end{proof}

\subsection{Proof of Lemma \ref{lemma:complex:path_kL}} \label{proof:complex:path_kL}

\begin{proof}
We have observed that $k_L$ cannot cross arcs in $\{(u_i,t)\}_{i \in I}$, $\{(s,v_j)\}_{j \in J}$, or $\{(s,v_j^N)\}_{j \in J}$. 
Therefore, it cannot use paths in $\{s \rightarrow u_i \rightarrow t\}_{i \in I} \cup \{s \rightarrow v_j \rightarrow t\}_{j \in J} \cup \{s \rightarrow v_j^N \rightarrow t\}_{j \in J} \cup \{s \rightarrow  v_j \rightarrow v_j^1 \rightarrow \dots \rightarrow v_j^N \rightarrow t\}_{j \in J}$.
Moreover, as established in the proof of Lemma \ref{lemma:complex:path_k2}, the delay on paths involving the arc chains $v_j \rightarrow v_j^1 \rightarrow \dots \rightarrow v_j^N$ is strictly greater than $d_L = 2 + \frac{1}{2W}$, which exceeds the delay limit for $k_L$. 
    
The only remaining feasible paths for $k_L$ are therefore those in $\{s \rightarrow u_i \rightarrow v_j \rightarrow t\}_{i \in I, j \in J}$.
\end{proof}

\subsection{Proof of Lemma \ref{lemma:complex:load_kL}} \label{proof:complex:load_kL}

\begin{proof}
From Lemma \ref{lemma:complex:path_kL}, commodity $k_L$ enters vertex $t$ exclusively through the set of arcs $\{(v_j,t)\}_{j \in J}$.
Equation \eqref{eq:complex:load_vjt} establishes that the total load on each of these arcs is  $\frac{W}{m} + 1$. Because commodity $k_2$ routes exactly $1$ unit of flow on each arc $(v_j,t)$, the remaining flow on these arcs must belong to commodity $k_L$. Therefore, the load of $k_L$ on each arc $(v_j,t)$ is exactly $\frac{W}{m}$. Summing over all $m$ arcs confirms that the total volume routed matches the demand $b_L = W$.
\end{proof}

\subsection{Proof of Lemma \ref{lemma:complex:load_kH}} \label{proof:complex:load_kH}

\begin{proof}
Equation \eqref{eq:complex:load_vjNt} establishes that the total load on each arc $(v_j^N, t)$ is  $W\frac{(n+1)(n-3)}{n} + 1$. Knowing that $1$ unit of flow is already routed by commodity $k_2$ on each arc $(v_j^N, t)$, the remaining flow on these arcs must belong to commodity $k_H$. Thus, the load of $k_H$ on each arc $(v_j^N, t)$ is $W\frac{(n+1)(n-3)}{n}$. 

By the conservation of flow, because commodity $k_H$ only enters $t$ through these arcs and must cross the entire arc chain to reach them, this result extends to all arcs in the chain: the entry arcs $\{(v_j, v_j^1)\}_{j \in J}$, the internal arcs $\{(v_j^l, v_j^{l+1})\}_{j \in J, l \in [1, N-1]}$, and the exit arcs $\{(v_j^N, t)\}_{j \in J}$. Finally, substituting these loads back into the delay functions confirms that the delay on each of these arcs is exactly $1$.
\end{proof}

\subsection{Proof of Lemma \ref{lemma:complex:single_arc_kL}} \label{proof:complex:single_arc_kL}

\begin{proof}
From Lemma \ref{lemma:complex:path_k2}, commodity $k_2$ never crosses the arcs $(u_i, v_j)$. These arcs are therefore exclusively used by commodities $k_L$ and $k_H$. From Lemma \ref{lemma:complex:load_k2}, we know that the delays of the entry and exit arcs are $d_{(s,u_i)} = 1$ and $d_{(v_j,t)} = 1$.

Lemma \ref{lemma:complex:load_kH} establishes that the total delay of the sub-path $v_j \rightarrow v_j^1 \rightarrow \dots \rightarrow v_j^N \rightarrow t$ used by commodity $k_H$ is $N+1$. 
Because the maximum delay limit for commodity $k_H$ is $d_H = N + 2 + \frac{1}{W}$, any arc $(u_i,v_j)$ traversed by $k_H$ must satisfy:
$$ d_{(s,u_i)} + d_{(u_i,v_j)} + (N+1) \leq d_H \implies 1 + d_{(u_i,v_j)} + N + 1 \leq N + 2 + \frac{1}{W} \implies d_{(u_i,v_j)} \leq \frac{1}{W}.$$
Since $d_{(u_i,v_j)} = \frac{1}{c_{(u_i,v_j)} -y_{(u_i,v_j)} }  = \frac{1}{w_i + 2W - y_{(u_i,v_j)}}$, this delay constraint imposes the upper bound on the load $y_{(u_i,v_j)} \leq w_i + W$ for any arc $(u_i,v_j)$ traversed by $k_H$. 

Similarly, the maximum delay limit for commodity $k_L$ is $d_L = 2 + \frac{1}{2W}$. Any arc $(u_i,v_j)$ traversed by $k_L$ must satisfy:
$$ d_{(s,u_i)} + d_{(u_i,v_j)} + d_{(v_j,t)} \leq d_L \implies 1 + d_{(u_i,v_j)} + 1 \leq 2 + \frac{1}{2W} \implies d_{(u_i,v_j)} \leq \frac{1}{2W}.$$
This stricter delay constraint imposes the tighter upper bound $y_{(u_i,v_j)} \leq w_i$ for any arc $(u_i,v_j)$ traversed by $k_L$.

\medskip

By flow conservation, the total flow exiting node $u_i$ equals the total flow entering that node for any $i \in I$. Lemma \ref{lemma:complex:load_k2} shows that the load entering $u_i$ is $y_{(s, u_i)} = m w_i + (m-1)W + 1$, and the load exiting on arc $(u_i,t)$ is $y_{(u_i,t)} = 1$. Therefore, the total load distributed across the $m$ outgoing arcs $\{(u_i,v_j)\}_{j \in J}$ is: $$\sum_{j \in J}y_{(u_i,v_j)}=y_{(s, u_i)}-1=m w_i + (m-1)W$$

We deduce that at least $m-1$ arcs must carry a load strictly greater than $w_i$. If this were not the case, at least two arcs would carry a load of at most $w_i$, and the remaining $m-2$ arcs would carry a maximum of $w_i + W$. The total flow would then be bounded by $2 w_i + (m-2)(w_i + W) = m w_i + (m-2)W$, which is strictly less than the flow $m w_i + (m-1)W$ exiting node $u_i$. 

Because commodity $k_L$ requires $y_{(u_i,v_j)} \leq w_i$, it can cross at most one arc $(u_i,v_j)$ for any given $i \in I$, meaning it can route a maximum volume of $w_i$ from that object node. Given that the total demand for commodity $k_L$ is $b_L = W = \sum_{i \in I} w_i$, it must maximize its allowed flow across all object nodes $\{u_i\}_{i \in I}$. Thus, $k_L$ must use exactly one arc $(u_i,v_j)$ per object $i \in I$ and route exactly $w_i$ units of flow across it. Consequently, the remaining $m-1$ arcs exiting each $u_i$ must carry exactly $w_i + W$ of commodity $k_H$ and exhibit a delay of exactly $\frac{1}{W}$.
\end{proof}

\subsection{Proof of Lemma \ref{lemma:complex:three_arcs_kL}} \label{proof:complex:three_arcs_kL}

\begin{proof}
Using the loads $y_{(s,v_j)} = 1$, $y_{(v_j,t)} = \frac{W}{m} + 1$, and $y_{(v_j,v_j^1)} = W\frac{(n+1)(n-3)}{n}$ for any $j \in J$ established in Lemmas \ref{lemma:complex:load_k2} and \ref{lemma:complex:load_kH} we have, by flow conservation at node $v_j$:
\begin{align*}
y_{(s,v_j)} + \sum_{i \in I} y_{(u_i,v_j)} &= y_{(v_j,t)} + y_{(v_j,v_j^1)}\\
1+\sum_{i \in I} y_{(u_i,v_j)} &= \left(\frac{W}{m} + 1\right) + W\frac{(n+1)(n-3)}{n} \\
\sum_{i \in I} y_{(u_i,v_j)} &= \frac{W}{m} + W\frac{(n+1)(n-3)}{n} = W \left( \frac{3}{n} + \frac{n^2 - 2n - 3}{n} \right) = W \left( \frac{n^2 - 2n}{n} \right)= W(n-2) \\
\sum_{i \in I} y_{(u_i,v_j)} 
&=  \sum_{i \in I} w_i + W(n-3).
\end{align*}

The only commodities crossing arcs $(u_i,v_j)$ are $k_L$ and $k_H$. Let $l$ be the number of arcs entering $v_j$ that carry commodity $k_L$; from Lemma \ref{lemma:complex:single_arc_kL},  commodity $k_H$ is carried by the remaining $n-l$ arcs. 
For every $i \in I$, the load of arc $(u_i,v_j)$ is exactly $w_i$ if it is crossed by commodity $k_L$, and exactly $w_i+W$ if it is crossed by commodity $k_H$. 

The total load on the $n$ arcs $(u_i,v_j)$ can be expressed as:
$$\sum_{i \in I} y_{(u_i,v_j)} = \sum_{\text{arcs carrying } k_L} w_i + \sum_{\text{arcs carrying } k_H} (w_i + W) = \sum_{i \in I} w_i + (n-l)W.$$

Hence, we have
$\sum_{i \in I} w_i + W(n-3) = \sum_{i \in I} w_i + W(n-l)$,
which implies $l=3$. Therefore, commodity $k_L$ is carried by exactly $3$ arcs in $\{(u_i,v_j)\}_{i \in I}$ for every $j \in J$.
\end{proof}

\subsection{Proof of Proposition \ref{prop:complex:flow_implies_part}} \label{proof:complex:flow_implies_part}

\begin{proof}
Lemma \ref{lemma:complex:single_arc_kL} establishes that for every $i \in I$, commodity $k_L$ is routed through exactly one arc $(u_i,v_j)$. For each $j \in J$, let $T_j \subset I$ denote the subset of objects $i$ for which commodity $k_L$ uses the arc $(u_i,v_j)$. Lemma \ref{lemma:complex:three_arcs_kL} proves that each subset $T_j$ contains exactly $3$ elements. Because every object $i \in I$ is assigned to exactly one subset, the collection $\{T_j\}_{j \in J}$ forms a partition of $I$ into triplets.

By flow conservation, the volume of commodity $k_L$ entering node $v_j$ must equal the volume exiting it. From Lemma \ref{lemma:complex:single_arc_kL}, the load of $k_L$ entering $v_j$ from the object nodes is $\sum_{i \in T_j} w_i$. From Lemma \ref{lemma:complex:load_kL}, the load of $k_L$ exiting $v_j$ on arc $(v_j, t)$ is $\frac{W}{m}$. Therefore, we have:
$ \sum_{i \in T_j} w_i = \frac{W}{m} \qquad \forall j \in J. $
We can conclude that the partition $\{T_j\}_{j \in J}$ is a valid solution for the \textsc{3-Partition} problem defined by the multiset $\{w_i\}_{i \in I}$.
\end{proof}

\section{Proof of the perfomance guarantee of $\mathcal{H}_{{threshold}}$}\label{sec:proof:perf_guaranttee}

\subsection{Proof of Proposition \ref{prop:approx:def_lambda}} \label{proof:approx:def_lambda}

\begin{proof}
Since $\check{f}_{\mathcal{X}_{\epsilon, \theta, \zeta, \eta}} (x,y) \geq s_{\epsilon, \theta}(x,y)$,  if $f(\lambda x, \lambda y) \leq \lambda s_{\epsilon, \theta}(x,y)$ holds, then  $f(\lambda x, \lambda y) \leq \lambda \check{f}_{\mathcal{X}_{\epsilon, \theta, \zeta, \eta}}(x,y)$ is also valid.

If $\lambda = 0$, the inequality in the proposition holds trivially. Let us then assume that  $\lambda > 0$. Inequality $f(\lambda x, \lambda y) \leq \lambda s_{\epsilon, \theta}(x,y)$ is then equivalent to: 
\begin{small_equation}
    \lambda \leq \frac{1-\frac{x}{s_{\epsilon,\theta}(x,y)}}{y}=\frac{1-\frac{x}{\frac{x}{1-\theta-\frac{x}{1-\frac{y-x-\theta}{1-\epsilon-\theta}}}}}{y}=\frac{(\theta+x)(1-\epsilon)-\theta y}{y(1-\epsilon+x)-y^2} = q_{\epsilon,\theta}(x,y). \label{eq:quality_q}
\end{small_equation}
For a given $y \in ]\eta,1-\epsilon]$, the function $q_{\epsilon,\theta}$ defined above is increasing in $x \in [0, \min(y-\theta,1-\zeta)]$ since its derivative is given by $\frac{(1 -\epsilon -y)(1-\epsilon - \theta)}{y(1 - \epsilon + x -y)^2}$. Similarly, $q_{\epsilon,\theta}$ is  increasing with respect to $\theta$ for any $\theta$ satisfying \eqref{eq:condi}.
Therefore, given an $x_0 \in [0, 1-\zeta]$, $f(\lambda x, \lambda y) \leq \lambda \check{f}(x,y)$ holds for any $x\geq x_0$ and any $\lambda \leq q_{\epsilon,0}(x_0, y)$.

Moreover, $q_{\epsilon,0}(x_0, y) = \frac{x_0 (1-\epsilon)}{y(1-\epsilon+x_0)-y^2}$ attains its lower bound when $y^*=\frac{1-\epsilon+x_0}{2}$, with a value of $\frac{4x_0(1-\epsilon)}{(1-\epsilon+x_0)^2}$. 
Hence $\lambda_{\epsilon, x_0}=q_{\epsilon,0}(x_0, y^*)=\frac{4x_0(1-\epsilon)}{(1-\epsilon+x_0)^2}$ satisfies \eqref{eq:quality_q}. 
\end{proof}

\subsection{Proof of Lemma \ref{lemma:approx:lambda_post}} \label{proof:approx:lambda_post}

\begin{proof}
First, observe that $\lambda_{1-\frac{u_a}{c_a}, \frac{b^k  \gamma}{c_a S |P^k|}} = 4 \frac{\frac{b^k \gamma}{c_a S |P^k|}\frac{u_a}{c_a}}{(\frac{u_a}{c_a}+\frac{b^k  \gamma}{c_a S |P^k|})^2}=\frac{4b^k \gamma~ u_a~S|P^k|}{(u_aS|P^k|+b^k \gamma)^2}$.
The derivative of this expression with respect to $b^k$, $\frac{\partial }{\partial b^k}\left(\lambda_{1-\frac{u_a}{c_a}, \frac{b^k  \gamma}{c_a S |P^k|}}\right) = \frac{4\gamma u_a S |P^k| (u_a S |P^k| - b^k \gamma)}{(u_a S |P^k| + b^k \gamma)^3}$, is positive for $b^k \leq \frac{u_a S |P^k|}{\gamma}$. Therefore, if $u_a \geq \frac{b^k \gamma}{S|P^k|}$,  $\lambda_{1-\frac{u_a}{c_a}, \frac{b^k  \gamma}{c_a S |P^k|}}$ increases with $b^k$ and $\lambda_{1-\frac{u_a}{c_a}, \frac{b^k  \gamma}{c_a S |P^k|}} \geq \frac{4b \gamma~ u_a~S|P^k|}{(u S|P^k|+b \gamma)^2}$ with $b = \min_{k \in K} b^k$. 

Similarly, the derivative of $\frac{4b \gamma~ u_a~S|P^k|}{(u_aS|P^k|+b \gamma)^2}$ with respect to $|P^k|$ is negative if $|P^k| \geq \frac{b\gamma}{S u_a}$. Notice that, under the assumption $u_a \ge \frac{b^k \gamma}{S |P^k|}$, it follows that $u_a \ge \frac{b \gamma}{S |P^k|}$.
 Therefore,   $\frac{4b \gamma~ u_a~S|P^k|}{(u_aS|P^k|+b \gamma)^2}$  decreases with $|P^k|$ and $\frac{4b \gamma~ u_a~S|P^k|}{(u_aS|P^k|+b \gamma)^2} \geq \frac{4b \gamma~ u_a~S|P|}{(u_aS|P|+b \gamma)^2}$ with  $|P|=\max_{k \in K} |P^k|$. 
 
Finally, the derivative of $\frac{4b \gamma~ u_a~S|P|}{(u_aS|P|+b \gamma)^2}$  with respect to $u_a$ is negative if $u_a \geq \frac{b\gamma}{S |P|}$. Hence, if $u_a \geq \frac{b^k \gamma}{ S |P^k|}$, $\frac{4b \gamma~ u_a~S|P|}{(u_aS|P|+b \gamma)^2}$ decreases with $u_a$ and $\frac{4b \gamma~ u_a~S|P|}{(u_aS|P|+b \gamma)^2} \geq \frac{4b \gamma~ u~S|P|}{(uS|P|+b \gamma)^2}$ with $u=\max_{a \in A} u_a$, and we get
$   \lambda_{1-\frac{u_a}{c_a}, \frac{b^k  \gamma}{c_a S |P^k|}} \geq \frac{4ub  \gamma S|P|}{(uS|P|+b  \gamma)^2 }$.
\end{proof}

\subsection{Proof of Proposition \ref{prop:approx:h_pos_guarantee_posterior}} \label{proof:approx:h_pos_guarantee_posterior}

\begin{proof}
Let us consider a scalar $S>1$ and a solution of $\mbox{DCMCF}_\text{env}$ ; variable values of this solution are denoted by a tilde ( $\tilde{\cdot}$ ). 
We define, for every commodity $k \in K$, the set of paths  $Q^k_S = \{p \in P^k: \tilde{x}^k_p \geq \frac{\tilde{\gamma}}{S|P^k|} \}$.

Let us construct a feasible solution ($\bar{x}^k_p, \bar{y}_a, \bar{\gamma})$ of DCMCF  as follows:   
\begin{small_align*}
    \bar{x}^k_p &=
    \begin{cases}
        \lambda^{post}_{\tilde{\gamma}, S} \tilde{x}^k_p & \text{if} \quad p \in Q^k_S\\
        0 & \text{if} \quad p \in P^k \setminus Q_S^k
    \end{cases} && \forall k \in K, p \in P^k \\
    \bar{y}_a &= \sum_{k \in K, p \ni a} b^k \bar{x}^k_p && \forall a \in A \\
    \bar{\gamma} &= \min_{k \in K} \sum_{p \in P^k} \bar{x}^k_p.
\end{small_align*}
If $p \in Q^k_S$, then $\bar x^k_p = {\lambda^{post}_{\tilde \gamma, S}} \tilde x^k_p$, and  $\bar y_a \leq {\lambda^{post}_{\tilde \gamma, S}} \tilde y_a$ for $a \in A$. 
Observe also that, for such a path $p$,  $\tilde{x}^k_p \geq \frac{\tilde{\gamma}}{S|P^k|}$ implies that  $u_a \geq \tilde{y}_a \geq b^k \tilde{x}^k_p \geq \frac{b^k \tilde{\gamma}}{S |P^k|}$ for any $a \in p$. One can then 
apply  Lemma \ref{lemma:approx:lambda_post} and Proposition \ref{prop:approx:def_lambda}  for every $k \in K$ and $p \in Q^k_S$:
\begin{small_equation*}
\sum_{a \in p} \frac{b^k \bar x^k_p}{c_a -  \bar y_a} \leq \sum_{a \in p}\frac{b^k {\lambda^{post}_{\tilde \gamma, S}} \tilde x^k_p}{c_a - {\lambda^{post}_{\tilde \gamma, S}} \tilde y_a} \leq {\lambda^{post}_{\tilde \gamma, S}} ~ \sum_{a \in p}\check f_{\mathcal{X}_{\frac{c_a-u_a}{c_a}, \frac{{l_a}^k_p}{c_a}, \frac{c_a - b^k u^k_p}{c_a}, \frac{l_a}{c_a}}}(\frac{b^k \tilde x^k_p}{c_a},\frac{\tilde y_a}{c_a})\leq {\lambda^{post}_{\tilde \gamma, S}} b^k d^k \tilde x^k_p = b^k d^k \bar{x}^k_p
\end{small_equation*}
proving that constraints \eqref{eq:dcmf:delay_time_xkp} are satisfied for all $k \in K$ and $p \in Q^k_S$. 

If $p \in P^k \setminus Q^k_S$ then $\bar x^k_p = 0$ and the constraint is trivially satisfied. 

It is straightforward to check that all other constraints of the DCMCF problem are satisfied, so the constructed solution  is  feasible. 

Since $S>1$, $Q^k_S$ contains at least one path.  Consequently,  the following lower bound holds for any $k \in K$:  
$
\sum_{p \in Q^k_S} \tilde{x}^k_p \geq \frac{S|P^k| - (|P^k|-1)}{S |P^k|} \tilde \gamma$, implying that    
\begin{small_align*}
    \bar \gamma &= \min_{k \in K} \sum_{p \in P^k} \bar x^k_p = \min_{k \in K} \sum_{p \in Q^k_S} {\lambda^{post}_{\tilde \gamma, S}}\tilde x^k_p \geq  \min_{k \in K} {\lambda^{post}_{\tilde \gamma, S}} \frac{S|P^k| - (|P^k|-1)}{S |P^k|}\tilde \gamma
\end{small_align*}
Observe that $\frac{\partial }{\partial |P^k|}\left(\frac{S|P^k| - (|P^k|-1)}{S |P^k|}\right) \leq 0$ on $\mathbb{R}^+$ ; therefore $\frac{S|P^k| - (|P^k|-1)}{S |P^k|} \geq \frac{S|P| - (|P|-1)}{S |P|}$ and $\bar \gamma \geq {\lambda^{post}_{\tilde \gamma, S}} \frac{S|P| - (|P|-1)}{S |P|}\tilde \gamma=\frac{4b\tilde \gamma~ u~S|P|}{(uS|P|+b\tilde \gamma)^2}\frac{S|P| - (|P|-1)}{S |P|}\tilde \gamma $.

Finally, let  $S^*=2-\frac{2}{|P|}+\frac{b \tilde \gamma}{u|P|}> 1$ be  the value of $S$  maximizing the lower bound of $\bar \gamma$. One can then deduce that:

\begin{small_align*}
\bar \gamma &\geq \frac{4b\tilde \gamma~ u~S^*|P|}{(uS^*|P|+b\tilde \gamma)^2} \frac{S^*|P| - (|P|-1)}{S^* |P|}\tilde \gamma= \frac{4b\tilde \gamma~ u (S^*|P| - (|P|-1))}{(uS^*|P|+b\tilde \gamma)^2} \tilde \gamma =   \frac{4b\tilde \gamma~ u \left((2-\frac{2}{|P|}+\frac{b \tilde\gamma}{u|P|})|P| - (|P|-1)\right)}{\left(u (2-\frac{2}{|P|}+\frac{b \tilde\gamma}{u|P|})|P|+b\tilde \gamma\right)^2} \tilde \gamma  = \frac{b \tilde \gamma} {u (|P|-1)+b\tilde \gamma}  \tilde \gamma.
\end{small_align*}
This shows that, given the solution of the $\mbox{DCMCF}_\text{env}$ problem with objective value $\tilde \gamma$, there exists a solution of  DCMCF  with objective value $\bar \gamma \geq \frac{b \tilde \gamma} {u (|P|-1)+b \tilde\gamma} \tilde \gamma$. 
\end{proof}

\subsection{Proof of Lemma \ref{lemma:approx:alpha}}\label{proof:approx:alpha}

\begin{proof}
Let us consider an instance of  DCMCF and a scalar $\alpha > 1$. We denote by $u_a^0$ the upper bound of variable $y_a$ for $a \in A$. Such bounds can be derived analytically  as described in Section \ref{sub:param}. Upper bounds of variables $x_p^k$ can similarly be computed for every $k \in K$ and $p \in P^k$, and lower bounds $l_a$ and ${l_a}_p^k$ can be set to zero.

We can now iteratively generate a sequence of relaxations $\{{DCMCF}_{env}^{~n}\}_{n \in \mathbb{N}}$ as follows. 
Given the relaxation ${DCMCF}_{env}^{~n}$, with upper bounds $u_a^{n}$  for each arc $a \in A$ and optimal objective value $\gamma_n^*$, either one of the two following conditions is met:
\begin{enumerate}
    \item \emph{Direct Success:} The relaxation immediately satisfies the desired property: for every arc $a \in A$, the optimal objective value of  ${DCMCF}_{env}^{~n}$ and the upper bound of every arc $a$ satisfy $u_a^{n} \leq \alpha  (\sum_{k \in K} b^k) \gamma^*_n$.
    \item \emph{Bounds Improvement:} The relaxation does not satisfy the desired property, but bounds can be improved: a new problem ${DCMCF}_{env}^{~n+1}$ with, for every arc $a \in A$, a new upper bound $u_a^{n+1}=\min \left( u_a^{n},  (\sum_{k \in K} b^k) \gamma^*_n \right)$ is created and solved. 
\end{enumerate}
The new problem ${DCMCF}_{env}^{~n+1}$ constructed in the latter case remains a relaxation of DCMCF, since $(\sum_{k \in K} b^k) \gamma_n^*$ is a valid upper bound on the load of each arc $a \in A$.

\medskip

If this process is run for $1+|A|$ iterations and the condition is still not satisfied (i.e., no direct success), then there exists at least one arc $a$ for which the upper bound has been updated twice. In other words, we have $u_a^n > \alpha (\sum_{k \in K} b^k) \gamma_n^*$ and $u_a^{n+l} > \alpha (\sum_{k \in K} b^k) \gamma_{n+l}^*$ for some $a \in A$, and some $l \leq 1+ |A|$. Hence, the following inequalities hold: $\alpha (\sum_{k \in K} b^k) \gamma_{n+l}^* < u_a^{n+l} \leq u_a^{n+1} \leq (\sum_{k \in K} b^k) \gamma_n^*$, implying that $\frac{\gamma_{n+l}^*}{\gamma_{n}^*} \leq \frac{1}{\alpha}$.  
Since the $\gamma_{n}^*$ values cannot fall below the lower bound   $\gamma_{low}$ defined in \eqref{eq:low5},  
 the iterative process must terminate and generate a relaxation of $\mbox{DCMCF}_\text{env}$ that satisfies the required conditions within a maximum of
$
(|A|+1) \left\lceil \log_{\alpha} \left( \frac{\gamma_0^*}{\gamma_{low}} \right) \right\rceil 
$ iterations. Since $\gamma_0^*$ is less than the trivial upper bound $\frac{\sum_{a \in A} c_a}{\sum_{k \in K} b^k}$, the number of iterations is polynomially bounded in the size of the instance, ending the proof.
\end{proof}

\subsection{Proof of Theorem \ref{theorem:approx:h_pos_guarantee_prior}} \label{proof:approx:h_pos_guarantee_prior}

\begin{proof}

Given a DCMCF instance and any constant $\alpha > 1$, Lemma~\ref{lemma:approx:alpha} ensures that one can construct a relaxation $\mathrm{DCMCF}_{\text{env}}$ with optimal value $\gamma^*$ such that, for every $a \in A$, 
$u_a \leq \alpha (\sum_{k \in K} b^k) \gamma^*$.
  
Proposition \ref{prop:approx:h_pos_guarantee_posterior}  shows that there exists a solution of  DCMCF with objective value $\bar \gamma^* \geq \frac{b \gamma^*}{u (|P|-1)+b\gamma^*} \gamma^*$.
Since $\frac{b \gamma^*}{u (|P|-1)+b\gamma^*} \gamma^*$ decreases with $u=\max_{a \in A} u_a \leq \alpha (\sum_{k \in K} b^k) \gamma^* \leq \alpha |K| \gamma^* \max_{k \in K}b^k$, we have: 
\begin{small_equation*}
    \bar \gamma^*  \geq \frac{b \gamma^*}{u (|P|-1)+b\gamma^*} \gamma^* \geq \frac{b \gamma^*}{\alpha |K| \gamma^*  (|P|-1)  \max_{k \in K} b^k+b\gamma^*} \gamma^* \geq \frac{1}{\alpha \frac{\max_{k \in K} b^k}{b} |K|(|P|-1)+1} \gamma^*=\frac{1}{\alpha \beta |K|(|P|-1)+1} \gamma^*
\end{small_equation*}
ending the  proof.
\end{proof}

\end{appendices}

\begin{singlespace}
\small
\bibliographystyle{abbrvnat}
\bibliography{biblio}

@article{yen1971,
  title={Finding the k shortest loopless paths in a network},
  author={Yen, Jin Y},
  journal={management Science},
  volume={17},
  number={11},
  pages={712--716},
  year={1971},
  publisher={Informs}
}

@article{garey1975complexity,
  title={Complexity results for multiprocessor scheduling under resource constraints},
  author={Garey, Michael R and Johnson, David S.},
  journal={SIAM journal on Computing},
  volume={4},
  number={4},
  pages={397--411},
  year={1975},
  publisher={SIAM}
}

@book{kleinrock1975theory,
  title={Theory, volume 1, queueing systems},
  author={Kleinrock, Leonard},
  year={1975},
  publisher={Wiley-interscience}
}

@article{mccormick1976computability,
  title={Computability of global solutions to factorable nonconvex programs: Part I—Convex underestimating problems},
  author={McCormick, Garth P},
  journal={Mathematical programming},
  volume={10},
  number={1},
  pages={147--175},
  year={1976},
  publisher={Springer}
}

@book{ahuja1988network,
  title={Network flows},
  author={Ahuja, Ravindra K and Magnanti, Thomas L and Orlin, James B},
  year={1988},
  publisher={Cambridge, Mass.: Alfred P. Sloan School of Management, Massachusetts}
}

@article{sherali1990explicit,
  title={An explicit characterization of the convex envelope of a bivariate bilinear function over special polytopes},
  author={Sherali, Hanif D and Alameddine, Amine},
  journal={Annals of Operations Research},
  volume={25},
  number={1},
  pages={197--209},
  year={1990},
  publisher={Springer}
}

@article{ouorou2000survey,
  title={A survey of algorithms for convex multicommodity flow problems},
  author={Ouorou, Adamou and Mahey, Philippe and Vial, J-Ph},
  journal={Management science},
  volume={46},
  number={1},
  pages={126--147},
  year={2000},
  publisher={INFORMS}
}

@inproceedings{yen2001,
  title={Near-optimal delay constrained routing in virtual circuit networks},
  author={Yen, Hong-Hsu and Lin, FY-S},
  booktitle={Proceedings IEEE INFOCOM 2001. Conference on Computer Communications. Twentieth Annual Joint Conference of the IEEE Computer and Communications Society (Cat. No. 01CH37213)},
  volume={2},
  pages={750--756},
  year={2001},
  organization={IEEE}
}

@inproceedings{beker2003,
  title={Off-line MPLS layout design and reconfiguration: Reducing complexity under dynamic traffic conditions},
  author={Beker, Sergio and Kofman, Daniel and Puech, Nicolas},
  booktitle={International Network Optimization Conference (INOC)},
  pages={61--66},
  year={2003}
}

@article{ben2006mathematical,
  title={Mathematical models of the delay constrained routing problem},
  author={Ben-Ameur, Walid and Ouorou, Adam},
  journal={Algorithmic Operations Research},
  volume={1},
  number={2},
  pages={94--103},
  year={2006},
  publisher={The Electronic Text Centre at the University of New Brunswick Libraries}
}

@techreport{duhamel2007augmented,
  title={An augmented lagrangean approach for the QoS constrained routing problem},
  author={Duhamel, Christophe and Mahul, Antoine},
  institution={LIMOS, UMR6158-CNRS},
  year={2007}
}

@article{wang2008,
  title={An overview of routing optimization for internet traffic engineering},
  author={Wang, Ning and Ho, Kin Hon and Pavlou, George and Howarth, Michael},
  journal={IEEE Communications Surveys \& Tutorials},
  volume={10},
  number={1},
  pages={36--56},
  year={2008},
  publisher={IEEE}
}

@article{truffot2010k-splittable,
  title={k-Splittable delay constrained routing problem: A branch-and-price approach},
  author={Truffot, J{\'e}r{\^o}me and Duhamel, Christophe and Mahey, Philippe},
  journal={Networks},
  volume={55},
  number={1},
  pages={33--45},
  year={2010},
  publisher={Wiley Online Library}
}

@article{hijazi2010outer,
  title={An outer-inner approximation for separable MINLPs},
  author={Hijazi, Hassan and Bonami, Pierre and Ouorou, Adam},
  journal={LIF},
  year={2010}
}

@article{orlowski2010sndlib,
  title={SNDlib 1.0—Survivable network design library},
  author={Orlowski, Sebastian and Wessaly, Roland and Pioro, Michal and Tomaszewski, Artur},
  journal={Networks: An International Journal},
  volume={55},
  number={3},
  pages={276--286},
  year={2010},
  publisher={Wiley Online Library}
}

@article{hijazi2012,
  title={Mixed-integer nonlinear programs featuring “on/off” constraints},
 journal={Computational Optimization and Applications},
  volume={52},
  number={2},
  pages={537--558},
  year={2012},
  publisher={Springer},
author={Hijazi, Hassan and Bonami, Pierre and Cornu{\'e}jols, G{\'e}rard and Ouorou, Adam},
}

@article{hijazi2013robust,
  title={Robust delay-constrained routing in telecommunications},
  author={Hijazi, Hassan and Bonami, Pierre and Ouorou, Adam},
  journal={Annals of Operations Research},
  volume={206},
  number={1},
  pages={163--181},
  year={2013},
  publisher={Springer}
}

@article{fortz2017,
  title={Models for the piecewise linear unsplittable multicommodity flow problems},
  author={Fortz, Bernard and Gouveia, Lu{\'\i}s and Joyce-Moniz, Martim},
  journal={European Journal of Operational Research},
  volume={261},
  number={1},
  pages={30--42},
  year={2017},
  publisher={Elsevier}
}

@article{bonami2017maximum,
  title={Maximum flow under proportional delay constraint},
  author={Bonami, Pierre and Mazauric, Dorian and Vax{\`e}s, Yann},
  journal={Theoretical Computer Science},
  volume={689},
  pages={58--66},
  year={2017},
  publisher={Elsevier}
}

@article{locatelli2018convex,
  title={Convex envelopes of bivariate functions through the solution of KKT systems},
  author={Locatelli, Marco},
  journal={Journal of global optimization},
  volume={72},
  number={2},
  pages={277--303},
  year={2018},
  publisher={Springer}
}

@article{papadimitriou2024augmented,
  title={An augmented Lagrangian method for nonconvex composite optimization problems with nonlinear constraints},
  author={Papadimitriou, Dimitri and V{\~u}, Bằng C{\^o}ng},
  journal={Optimization and Engineering},
  volume={25},
  number={4},
  pages={1921--1990},
  year={2024},
  publisher={Springer}
}

@inproceedings{dong2020tina,
  title={TINA: A fair inter-datacenter transmission mechanism with deadline guarantee},
  author={Dong, Xiaodong and Li, Wenxin and Zhou, Xiaobo and Li, Keqiu and Qi, Heng},
  booktitle={IEEE INFOCOM 2020-IEEE Conference on Computer Communications},
  pages={2017--2025},
  year={2020},
  organization={IEEE}
}

@article{beraudsudreau2026multi,
  title={On the Multi-Commodity Flow With Convex Objective Function: Column-Generation Approaches},
  author={Beraud-Sudreau, Guillaume and L{\'e}tocart, Lucas and Magnouche, Youcef and Martin, S{\'e}bastien},
  journal={Networks},
  year={2026},
  publisher={Wiley Online Library}
}

@TECHREPORT{B01,
  author = {D. Bienstock},
  title = {Potential function methods for approximatively solving linear programs: theory and practice},
  institution = {CORE Lecture Series Monograph},
  year = {2001}
}

@ARTICLE{BR02,
  author = {D. Bienstock and O. Raskina},
  title = {Asymptotic analysis of the flow deviation method for the maximum concurrent flow problem},
  journal = {Mathematical Programming, Ser. B},
  year = {2002},
  volume = {91},
  pages = {479--492}
}

@ARTICLE{LR99,
  author = {T. Leighton and S. Rao},
  title = {Multicommodity max-flow min-cut theorems and their use in designing approximation algorithms},
  journal = {JACM},
  year = {1999},
  volume = {46},
  pages = {215--245}
}

@TECHREPORT{MS86,
  author = {D.W. Matula and F. Shahrokhi},
  title = {The maximum concurrent flow problem and sparsest cuts},
  institution = {Southern Methodist Univ. Dallas},
  year = {1986}
}

@INCOLLECTION{S97,
  author = {D.B. Shmoys},
  title = {Cut Problems and Their Application to Divide-and-Conquer},
  booktitle = {Approximation Algorithms for NP-Hard Problems},
  publisher = {PWS Publishing Company},
  editor = {D.S. Hochbaum},
  year = {1997},
  pages = {192--235}
}

@ARTICLE{KAR08,
  author = {G. Karakostas},
  title = {Faster approximation schemes for fractional multicommodity flow problems},
  journal = {ACM Transactions on Algorithms},
  year = {2008},
  volume = {4},
  number={1},
  pages                = {1-17}
}

@article{GargK07,
  author    = {N. Garg and
               J. K{\"o}nemann},
  title     = {Faster and Simpler Algorithms for Multicommodity Flow and
               Other Fractional Packing Problems},
  journal   = {SIAM J. Comput.},
  volume    = {37},
  number    = {2},
  year      = {2007},
  pages     = {630-652}
}

@article{Fleischer00,
  author    = {L. Fleischer},
  title     = {Approximating Fractional Multicommodity Flow Independent
               of the Number of Commodities},
  journal   = {SIAM J. Discrete Math.},
  volume    = {13},
  number    = {4},
  year      = {2000},
  pages     = {505-520},
 }

@article{Bauguion,
author = {Bauguion, Pierre-Olivier and Ben-Ameur, Walid and Gourdin, Eric},
title = {Efficient algorithms for the maximum concurrent flow problem},
journal = {Networks},
volume = {65},
number = {1},
pages = {56-67},
doi = {https://doi.org/10.1002/net.21572},
url = {https://onlinelibrary.wiley.com/doi/abs/10.1002/net.21572},
eprint = {https://onlinelibrary.wiley.com/doi/pdf/10.1002/net.21572},
year = {2015}
}

@book{Pioro,
  author       = {Michal Pi{\'{o}}ro and
                  Deepankar Medhi},
  title        = {Routing, flow, and capacity design in communication and computer networks},
  publisher    = {Morgan Kaufmann},
  year         = {2004},
  isbn         = {978-0-12-557189-0},
  bibsource    = {dblp computer science bibliography, https://dblp.org}
}

@article{Loca16,
	Author = {Locatelli, Marco},
	Da = {2016/12/01},
	Journal = {Journal of Global Optimization},
	Number = {4},
	Pages = {629--668},
	Title = {Polyhedral subdivisions and functional forms for the convex envelopes of bilinear, fractional and other bivariate functions over general polytopes},
	Ty = {JOUR},
	Volume = {66},
	Year = {2016}}

@article{2025scip,
  title={The SCIP Optimization Suite 10.0},
  author={Hojny, Christopher and Besan{\c{c}}on, Mathieu and Bestuzheva, Ksenia and Borst, Sander and Dion{\'\i}sio, Jo{\~a}o and Ehls, Johannes and Eifler, Leon and Ghannam, Mohammed and Gleixner, Ambros and G{\"o}{\ss}, Adrian and others},
  journal={arXiv preprint arXiv:2511.18580},
  year={2025}
}
%\printbibliography
\end{singlespace}
\end{document}